\documentclass[preprint,12pt]{elsarticle}

\usepackage{amssymb}
\usepackage{amsthm}

\usepackage{comment}
\usepackage{graphicx,amssymb,amsmath}
\usepackage{gensymb}
\usepackage{array,booktabs,enumitem}% http://ctan.org/pkg/{array,booktabs,enumitem}
\usepackage{color, colortbl}
\definecolor{Gray}{gray}{0.95}

\makeatletter
\def\corref#1{\edef\cnotenum{\elsRef{#1}}%
	\edef\@corref{\ifcase\cnotenum\or
		$\dagger$\or$\dagger\dagger$\fi\hskip-1pt}}

\def\cortext[#1]#2{\g@addto@macro\@cornotes{%
		\refstepcounter{cnote}\elsLabel{#1}%
		\def\thefootnote{\ifcase\thecnote\or$\dagger$\or
			$\dagger\dagger$\fi}%
		\footnotetext{#2}}}
\makeatother

\makeatletter
\def\ps@pprintTitle{%
	\let\@oddhead\@empty
	\let\@evenhead\@empty
	\def\@oddfoot{}%
	\let\@evenfoot\@oddfoot}
\makeatother

\usepackage[explicit]{titlesec}
\titleformat{\paragraph}[runin]{\normalfont\bfseries\itshape}{\theparagraph}{1em}{#1.}

\graphicspath{{./figs/}}

\journal{}

\usepackage{thmtools}
\usepackage{thm-restate}

\newtheorem{thm}{Theorem}[section]
\newtheorem{lem}[thm]{Lemma}
\newtheorem{cor}[thm]{Corollary}
\newtheorem{prob}{Problem}%[section]
\newtheorem{subprob}{Subproblem}
\newtheorem{obs}[thm]{Observation}
\newtheorem*{remark}{Remark}

\begin{document}

\begin{frontmatter}
	
	%% Title, authors and addresses
	
	%% use the tnoteref command within \title for footnotes;
	%% use the tnotetext command for theassociated footnote;
	%% use the fnref command within \author or \address for footnotes;
	%% use the fntext command for theassociated footnote;
	%% use the corref command within \author for corresponding author footnotes;
	%% use the cortext command for theassociated footnote;
	%% use the ead command for the email address,
	%% and the form \ead[url] for the home page:
	%% \title{Title\tnoteref{label1}}
	%% \tnotetext[label1]{}
	%% \author{Name\corref{cor1}\fnref{label2}}
	%% \ead{email address}
	%% \ead[url]{home page}
	%% \fntext[label2]{}
	%% \cortext[cor1]{}
	%% \address{Address\fnref{label3}}
	%% \fntext[label3]{}
	
	\title{Characterization and Computation of Feasible Trajectories for an Articulated Probe with a Variable-Length End Segment\tnoteref{label1}}
	\tnotetext[label1]{A preliminary version of this work was presented at the 32nd Annual Canadian Conference on Computational Geometry.}
	
	%% use optional labels to link authors explicitly to addresses:
	%% \author[label1,label2]{}
	%% \address[label1]{}
	%% \address[label2]{}
	
	\author{Ovidiu Daescu}
	\ead{ovidiu.daescu@utdallas.edu}
	
	\author{Ka Yaw Teo\corref{cor1}}
	\ead{ka.teo@utdallas.edu}
	
	\cortext[cor1]{Corresponding author}
	\address{Department of Computer Science, University of Texas at Dallas, Richardson, TX, USA.}

\begin{abstract}
An articulated probe is modeled in the plane as two line segments, $ab$ and $bc$, joined at $b$, with $ab$ being very long, and $bc$ of some small length $r$.
We investigate a trajectory planning problem involving the articulated two-segment probe where the length $r$ of $bc$ can be customized.
Consider a set $P$ of simple polygonal obstacles with a total of $n$ vertices, a target point $t$ located in the free space such that $t$ cannot see to infinity, and a circle $S$ centered at $t$ enclosing $P$.
The probe initially resides outside $S$, with $ab$ and $bc$ being collinear, and is restricted to the following sequence of moves: a straight line insertion of $abc$ into $S$ followed by a rotation of $bc$ around $b$.
The goal is to compute a feasible obstacle-avoiding trajectory for the probe so that, after the sequence of moves, $c$ coincides with $t$.

%We consider an extension of the articulated probe trajectory planning problem introduced in \cite{teo20traj}, where the length $r$ of the end segment can be customized.
We prove that, for $n$ line segment obstacles, the smallest length $r$ for which there exists a feasible probe trajectory can be found in $O(n^{2+\epsilon})$ time using $O(n^{2+\epsilon})$ space, for any constant $\epsilon > 0$.
Furthermore, we prove that all values $r$ for which a feasible probe trajectory exists form $O(n^2)$ intervals, and can be computed in $O(n^{5/2})$ time using $O(n^{2+\epsilon})$ space.
We also show that, for a given $r$, the feasible trajectory space of the articulated probe can be characterized by a simple arrangement of complexity $O(n^2)$, which can be constructed in $O(n^2)$ time.
To obtain our solutions, we design efficient data structures for a number of interesting variants of geometric intersection and emptiness query problems.
\end{abstract}

\begin{comment}
\begin{keyword}
	articulated probe trajectory \sep
	motion planning \sep	
	circular sector emptiness \sep
	circular arc shooting \sep
	radius intersection \sep
	radius shooting
\end{keyword}
\end{comment}

\end{frontmatter}

\section{Introduction}
\label{sec:intro}

We consider the \emph{articulated probe trajectory planning problem} originally introduced by Teo, Daescu, and Fox \cite{teo20traj} with the following setup.
We are given a two-dimensional workspace containing a set $P$ of simple polygonal obstacles with a total of $n$ vertices, and a target point $t$ in the free space, all enclosed by a circle $S$ of radius $R$ centered at $t$.
An articulated probe is modeled in $\Re^2$ as two line segments, $ab$ and $bc$, connected at point $b$.
The length of $ab$ is greater than or equal to $R$, whereas $bc$ is of some small length $r \in (0, R]$.
The probe is initially located outside $S$, assuming an \emph{unarticulated} configuration, in which $ab$ and $bc$ are collinear, and $b \in ac$.
A \emph{feasible probe trajectory} consists of an initial insertion (sliding) of straight line segment $abc$ into $S$, possibly followed by a rotation of $bc$ around $b$ up to $\pi/2$ radians in either direction, such that $c$ coincides with $t$, while avoiding the obstacles in the process.
If a rotation is performed, then we have an \emph{articulated final} configuration of the probe.
The objective of the problem is to determine if a feasible probe trajectory exists and, if so, to report (at least) one such trajectory. 

\begin{figure}[h]
	\centering
	\includegraphics[scale=0.16]{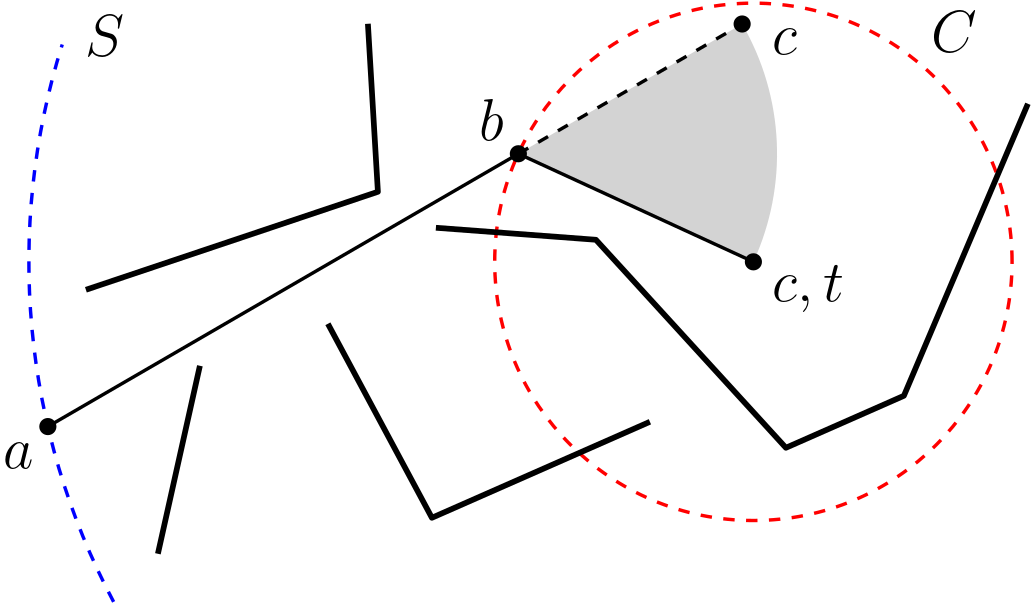}
	\caption{In order for the articulated probe to reach the target point $t$, a straight insertion of line segment $abc$ may be followed by a rotation of $bc$ from its intermediate position (dashed line segment) to its final position (solid line segment).}
	\label{fig_traj}
\end{figure}

It has been previously argued by Teo, Daescu, and Fox \cite{teo20traj} that the polygonal obstacles can be treated by considering only their bounding line segments.
Thus, for simplicity, assume that $P$ consists of $n$ non-crossing line segment obstacles (Figure \ref{fig_traj}).
We further assume that $S$ is not visible from $t$, since otherwise the problem reduces to computing visibility from $t$ to infinity, which takes $O(n \log n)$ time \cite{teo20traj}.
Thus, in order to reach $t$, the probe has to rotate $bc$ around $b$.

After inserting line segment $abc$, point $a$ is located either on or outside $S$.
Let $C$ be the circle of radius $r$ centered at $t$.
Observe that, since $bc$ may only rotate as far as $\pi/2$ radians in either direction after the initial insertion of line segment $abc$, $ab$ intersects $C$ only once and at $b$ (i.e., $b \in C$).
Thus, $ab$ never enters $C$ at all.
When $bc$ rotates around $b$, the area swept by $bc$ is a sector of a circle of radius $r$ centered at $b$.
For conciseness, the center of the circle on which a circular sector is based is called the \emph{center of the circular sector}.

%Suppose that the length $r$ of $bc$ is \emph{customizable} -- that is, the value $r$ may be chosen from the interval $(0, R]$.
In this paper, we develop efficient algorithms for computing
i) the minimum length $r > 0$ for which a feasible articulated trajectory exists, including reporting at least one such trajectory,
ii) all lengths $r > 0$ for which a feasible articulated trajectory exists, which we refer to as the \emph{feasible domain} of $r$, and
iii) the feasible trajectory space (i.e., set of all feasible trajectories) for a given length $r$.

%Note that when $r = 0$, we have the case in which the probe assumes an unarticulated final configuration.
%Finding the set of feasible unarticulated probe trajectories (and thus detecting the existence of a feasible unarticulated probe trajectory) essentially reduces to a radial visibility problem (i.e., radial visiblity from $t$ to infinity), which can be solved in $O(n \log n)$ time \cite{teo20traj}.
%Hereafter, we may simply assume that an unarticulated probe trajectory does not exist.

\subsection*{Related work.}
A \emph{linkage} is a sequence of fixed-length line segments connected consecutively through their endpoints.
The motion of a linkage has been previously studied for its geometric and topological properties \cite{connelly17geom,hopcroft84mov} as well as for its application in robotic arm modeling and motion planning \cite{choset05principles,lavalle06plan}.
Inexact approaches based on sampling \cite{kavraki96probabilistic,lavalle01randomized} and subdivision \cite{brooks85subdivision,donald84motion,zhu90constraint} are commonly employed for solving motion planning problems involving complex robots such as linkages with high-dimensional configuration spaces.
However, a sampling-based method is probabilistically complete at best, while a subdivision-based method is complete only at a chosen resolution.
These approximate methods are overall effective but not in planning manipulative tasks for multi-link robots with spatial constraints on links and joints \cite{yakey01randomized}.

Unlike polygonal linkages that can rotate freely at their joints while moving between a start and target configuration \cite{connelly17geom,hopcroft84mov,lavalle06plan}, our simple articulated probe is constrained to a fixed sequence of moves -- a straight line insertion, possibly followed by a rotation of the end link.
This type of linkage motion has not received attention until recently \cite{teo20traj,daescu19traj3d}.

The two-dimensional articulated probe trajectory planning problem with a constant length $r$ was introduced by Teo, Daescu, and Fox \cite{teo20traj}, who presented a geometric-combinatorial algorithm for computing so-called \emph{extremal} feasible probe trajectories in $O(n^2 \log n)$ time using $O(n \log n)$ space.
In an extremal probe trajectory, one or two obstacle endpoints always lie tangent to the probe.
The solution approach proposed in \cite{teo20traj} can be extended to the case of polygonal obstacles.
For $h$ polygonal obstacles with a total of $n$ vertices, an extremal feasible probe trajectory can be determined in $O(n^2 + h^2 \log h)$ time using $O(n \log n)$ space.
When a clearance $\delta$ from the polygonal obstacles is required, a feasible probe trajectory can be obtained in $O(n^2 + h^2 \log h)$ time using $O(n^2)$ space.

In addition, Daescu and Teo \cite{daescu19traj3d} developed an algorithm for solving the articulated probe trajectory planning problem in three dimensions for a given $r$.
It was shown that a feasible probe trajectory among $n$ triangular obstacles can be found in $O(n^{4+\epsilon})$ time using $O(n^{4+\epsilon})$ space, for any constant $\epsilon > 0$.

\subsection*{Motivation.}
The proposed trajectory planning problem, aside from its general relevance in robotics, arises specifically in some current medical applications.
In minimally invasive surgeries, a rigid needle-like instrument is typically inserted through a small incision to reach a given target, after which it may perform operations such as tissue resection and biopsy.
Some newer designs allow for a joint to be incorporated for moving the acting end (tip); after inserting the instrument in a straight path, the surgeon may rotate the tip around the joint to reach the target \cite{simaan18medical}.

Due to the rapid advances in three dimensional printing techniques, such robotic probes can even be customized for a given patient \cite{culmone19additive}.
Rather than using a one-size-fits-all instrument, based on the patient-specific requirement and constraints, a robotic probe with a tailored-sized tip can be customarily built on-demand using three dimensional printing.
These recent developments in surgical probes have enhanced the ability to reach obscured targets, improved the operational accuracy, and reduced the cost, error, and time duration of surgeries, but have also greatly increased the complexity of finding acceptable trajectories.

Despite its importance and relevance, as well as its rich combinatorial and geometric properties, only a handful of results have been reported \cite{teo20traj,daescu19traj3d} for the articulated probe trajectory planning problem. 
It may seem inevitable, due to practicality, to use heuristics and approximation; nonetheless, an exact solution approach exposes the rich combinatorial and geometric properties of the problem, whose exploitation has often proven necessary for achieving algorithmic improvement.

\subsection*{Results and contributions.}
Recall our assumption that there is no feasible unarticulated probe trajectory (i.e., $t$ cannot see to infinity).
We begin in Section \ref{sec:min_r} by addressing our first problem of interest:

\begin{prob}
	\label{prob_min_r}
	Find the minimum length $r > 0$ of line segment $bc$ such that a feasible articulated probe trajectory exists, and report (at least) one such trajectory, or report that no feasible solution exists.
\end{prob}

For brevity, a feasible articulated trajectory with the minimum length $r$ is referred to as a \emph{feasible min-$r$ articulated trajectory}.

Our approach to solving Problem \ref{prob_min_r} is as follows:
i) We show that a feasible min-$r$ articulated trajectory, if one exists, can always be perturbed, while remaining feasible, into one of a finite number of ``extremal'' feasible trajectories, which can be enumerated using an algebraic-geometric method (see Lemma \ref{lem1} for a detailed definition of the extremal trajectories).
This leads to a simple $O(n^3 \log n)$ time, $O(n^{2+\epsilon})$ space algorithm, for any constant $\epsilon > 0$, based on enumerating and verifying the extremal trajectories for feasibility.
ii) We then derive an $O(n^{2+\epsilon})$ time and space algorithm by partially waiving the notion of computing and checking the extremal trajectories for feasibility.
Specifically, the algorithm searches for a feasible min-$r$ articulated trajectory, if any, by performing a finite sequence of perturbations and feasibility tests on certain $O(n^2)$ extremal trajectories (whose line segment $ab$ is tangent to two obstacle endpoints).
As it happens, this solution approach can be extended to solve our second problem of interest:

\begin{prob}
\label{prob_all_r}
Find all lengths $r \in (0,R]$ of line segment $bc$ for which at least one feasible articulated probe trajectory exists.
\end{prob}

As we shall see in Section \ref{sec:all_r}, with a proper algorithmic extension to our process of finding the minimum feasible $r$, we can compute, in $O(n^{5/2})$ time using $O(n^{2+\epsilon})$ space, the set of $r$-intervals for which feasible articulated trajectories exist, together with an implicit representation of feasible solutions for those values of $r$.

In the process of deriving our solutions to Problems \ref{prob_min_r} and \ref{prob_all_r}, we encounter and solve a number of fundamental problems (or their special cases) that could be of theoretical interest in computational geometry.
For instance, we provide an efficient data structure with logarithmic query time for solving a special instance of the circular sector emptiness query problem (i.e., for a query circular sector with a fixed arc endpoint $t$).
%and the radius intersection query problem (i.e., for a query segment with a fixed endpoint $t$).

At last, in Section \ref{sec:feas_space}, we proceed to our third problem:

\begin{prob}
	\label{prob_traj_space}
	For a given length $r$ of line segment $bc$, compute the feasible trajectory space (i.e., set of all feasible trajectories) of the articulated probe.
\end{prob}

We describe a geometric combinatorial approach for characterizing and computing the feasible trajectory space of the articulated probe.
The feasible configuration space has a worst-case complexity of $O(n^2)$ and can be described by an arrangement of simple curves.
Using topological sweep \cite{balaban95optimal}, the arrangement can be constructed in $O(n \log n + k)$ time using $O(n + k)$ working storage, where $k=O(n^2)$ is the number of vertices of the arrangement.
By simply traversing the cells of the arrangement, we can find a feasible probe trajectory in $O(n^2)$ time -- a logarithmic factor improvement compared to the algorithm in \cite{teo20traj}.

\section{Computing feasible min-\boldmath$r$ articulated trajectories}
\label{sec:min_r}

Recall that, for a given $r$, $C$ is the circle of radius $r$ centered at $t$.
%With $r$ being a variable (instead of a fixed constant), an articulated probe trajectory cannot be ``isolated'' by its intersections with two obstacle endpoints as described in \cite{teo20traj}.
Using the rationale of \cite[Lemma 2.1]{teo20traj}, we can immediately claim the following observation.

\begin{obs}
	\label{obs1}
	Given a feasible min-$r$ articulated trajectory, there exists an \textbf{extremal} feasible min-$r$ articulated trajectory such that the probe assumes an articulated final configuration that passes through an obstacle endpoint outside $C$ and another obstacle endpoint inside or outside $C$.
\end{obs}

We will later show in Lemma \ref{lem1} that an extremal feasible min-$r$ articulated trajectory is always tangent to two obstacle endpoints \emph{outside} $C$. 

For ease of discussion, unless noted otherwise, we use $bc$ and $bt$ to denote line segment $bc$ of the probe in its intermediate (right after the initial insertion of line segment $abc$) and final configurations, respectively.
Let $\angle cbt$ be the angle of rotation of line segment $bc$ (of the probe) to reach $t$, and 
let $\sigma_{bct}$ be the circular sector swept by the said rotation of $bc$.
Let $\gamma_{ct}$ denote the circular arc of $\sigma_{bct}$.
Let $V$ denote the set of endpoints of the line segments of $P$.

%The following observation is critical to our solution approach, as it implies that we only have to examine a finite set of ``extremal'' probe trajectories in order to find a feasible min-$r$ articulated trajectory, if one exists.

\begin{lem}
	\label{lem1}
	Given a feasible min-$r$ articulated trajectory, there exists an \textbf{extremal} feasible min-$r$ articulated trajectory such that, in its final configuration, $ab$ passes through two obstacle endpoints and at least one of the following holds:
	%\addtolength\leftmargini{0.5em}
	\renewcommand\labelenumi{\Roman{enumi})}
	\renewcommand\theenumi\labelenumi
	\begin{enumerate}
	\item $\angle cbt = \pi/2$ radians,
	\item $bc$ intersects an obstacle line segment at $c$,
	\item $\gamma_{ct}$ intersects an obstacle endpoint or is tangent to an obstacle line segment,
	\item one of the obstacle endpoints intersected by $ab$ coincides with $b$, and $\angle cbt \leq \pi/2$ radians, or
	\item $bt$ passes through an obstacle endpoint.
	\end{enumerate}	
\end{lem}

\begin{proof}
\label{lem1_proof}

We proceed by considering the two possible scenarios implied by Observation \ref{obs1}.

\paragraph{Scenario A}
A feasible min-$r$ articulated probe trajectory exists such that $ab$ of the trajectory passes through two obstacles endpoints $u, v \in V$, where $u \neq v$. Obviously, $ab$ does not intersect the interior of any line segment of $P$.
Without loss of generality, assume that line segment $bc$ of the probe is rotated clockwise around $b$ to reach $t$ (the other case can be handled symmetrically), and $ab$ passes through $u$ and $v$ in the way depicted in Figure \ref{fig_min-r_1}.

\begin{figure}[h]
	\centering
	\includegraphics[scale=0.16]{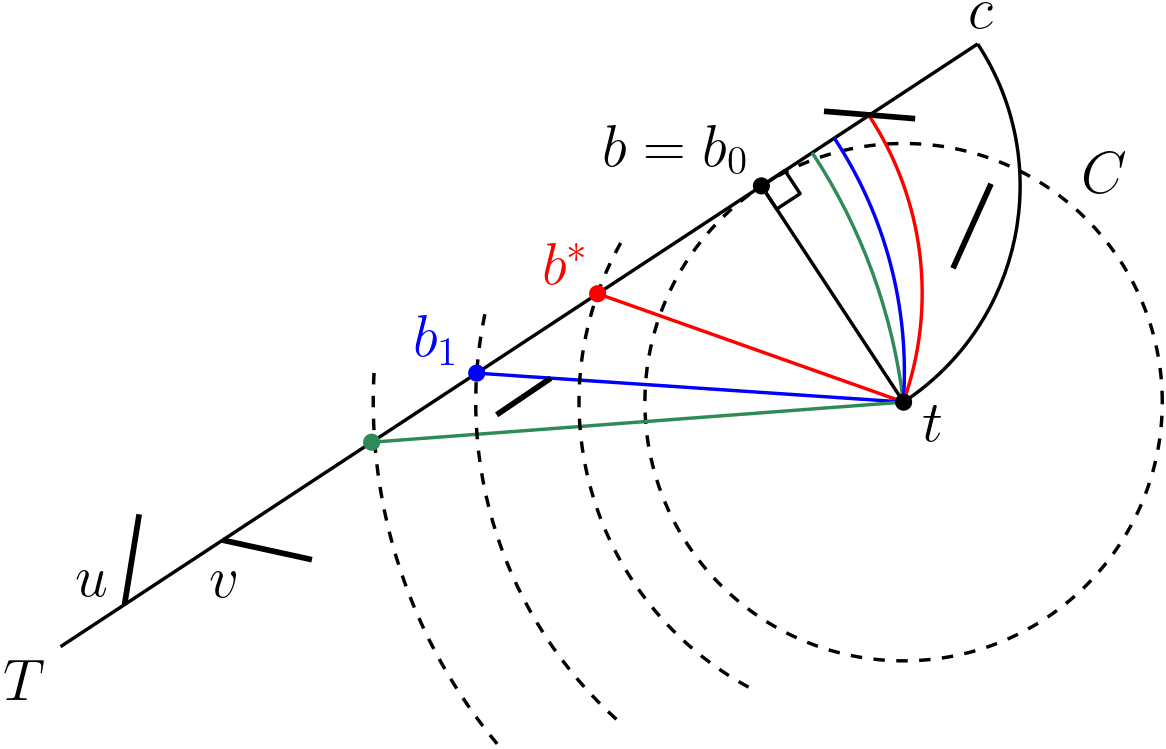}
	\caption{Finding an extremal feasible min-$r$ articulated probe trajectory in Scenario A.}
	\label{fig_min-r_1}
\end{figure}

Let $h_{ab}$ denote the supporting line of $ab$.
%Given its intersections with $u$ and $v$, $h_{ab}$ is fixed.
Let $b_0t$ be the perpendicular line segment dropped from $t$ to line $h_{ab}$.
Observe that the minimum possible value of $r$ for an articulated trajectory is given by the length of $b_0t$ -- that is, when $b = b_0$ and $\angle cbt$ is equal to $\pi/2$ radians.
Let $T$ denote the corresponding trajectory.
If $T$ is free of obstacles, then $T$ is a feasible min-$r$ articulated trajectory 
%that intersects $u$ and $v$ outside $C$
(case I of the lemma).

Otherwise, the minimum feasible value of $r$ is attained at some point $b^*$ on line segment $vb_0$, where $b^*$ is the closest point to $b_0$ on $vb_0$ for which the corresponding articulated trajectory is feasible.
In order to find $b^*$, we increase $r$ by moving $b$ away from $b_0$ on $vb_0$ until the trajectory becomes feasible.
Observe that,
%\renewcommand\labelenumi{\roman{enumi})}
%\renewcommand\theenumi\labelenumi
%\begin{enumerate}
%	\item If $ab$ of $T$ intersects any obstacle line segment, then there is certainly no feasible articulated trajectory that intersects $u$ and $v$ outside $C$.
%Thus, we can conclude that $ab$ of $T$ does not intersect any obstacle line segment.
if $bt$ intersects an obstacle line segment at any given time during the process of increasing $r$,
%by moving $b$ away from $b_0$ on $vb_0$, 
then the trajectory would never become feasible thereafter (illustrated by the blue and green trajectories in Figure \ref{fig_min-r_1}).
%\end{enumerate}

The observations above imply that, if $b = b_0$ is not feasible, then $bc$ or $\gamma_{ct}$ of $T$ must be intersected by an obstacle line segment, or $\sigma_{bct}$ of $T$ must contain an obstacle line segment.
By moving $b$ away from $b_0$ on $vb_0$, we may rid the trajectory of obstacle line segments that intersect $bc$, $\gamma_{ct}$, or are contained within $\sigma_{bct}$.
Suppose that we increase $r$ until either $bt$ becomes tangent to an obstacle line segment or $b$ reaches $v$.
Let $b_1$ denote the final position of $b$.
Observe that $b^*$ must lie somewhere between $b_0$ and $b_1$, and as we increase $r$, $b = b^*$ when $bc$ intersects an obstacle line segment at $c$, or $\gamma_{ct}$ intersects an obstacle endpoint or is tangent to an obstacle line segment (cases II and III of the lemma).

\begin{remark}
Let $r_{b_0}$, $r_{b^*}$, and $r_{b_1}$ be the lengths of $bc$ when $b = b_0$, $b = b^*$, and $b = b_1$, respectively, where $r_{b_0} \leq r_{b^*} \leq r_{b_1}$.
Observe that $[r_{b^*}, r_{b_1}]$ is a feasible contiguous subset of $[r_{b_0}, r_{b_1}]$.
Indeed, based on the observations made thus far, it follows that, in Scenario A, there exists at most one contiguous feasible subset of $[r_{b_0}, r_{b_1}]$.
\end{remark}

\paragraph{Scenario B}
A feasible min-$r$ articulated probe trajectory exists such that $ab$ of the trajectory passes through an obstacle endpoint $u$, and $bt$ of the trajectory passes through an obstacle endpoint $v$, where $u, v \in V$ and $u \neq v$.
Recall that $\angle cbt$ of the trajectory is less than or equal to $\pi/2$ radians.
%, since segment $bc$ of the probe may only be rotated around $b$ up to $\pi/2$ radians.
Without loss of generality, assume that segment $bc$ of the probe is rotated clockwise around $b$ to reach $t$, as in Figure \ref{fig_min-r_2} (the other case is symmetrical).

\begin{figure}[h]
	\centering
	\includegraphics[scale=0.16]{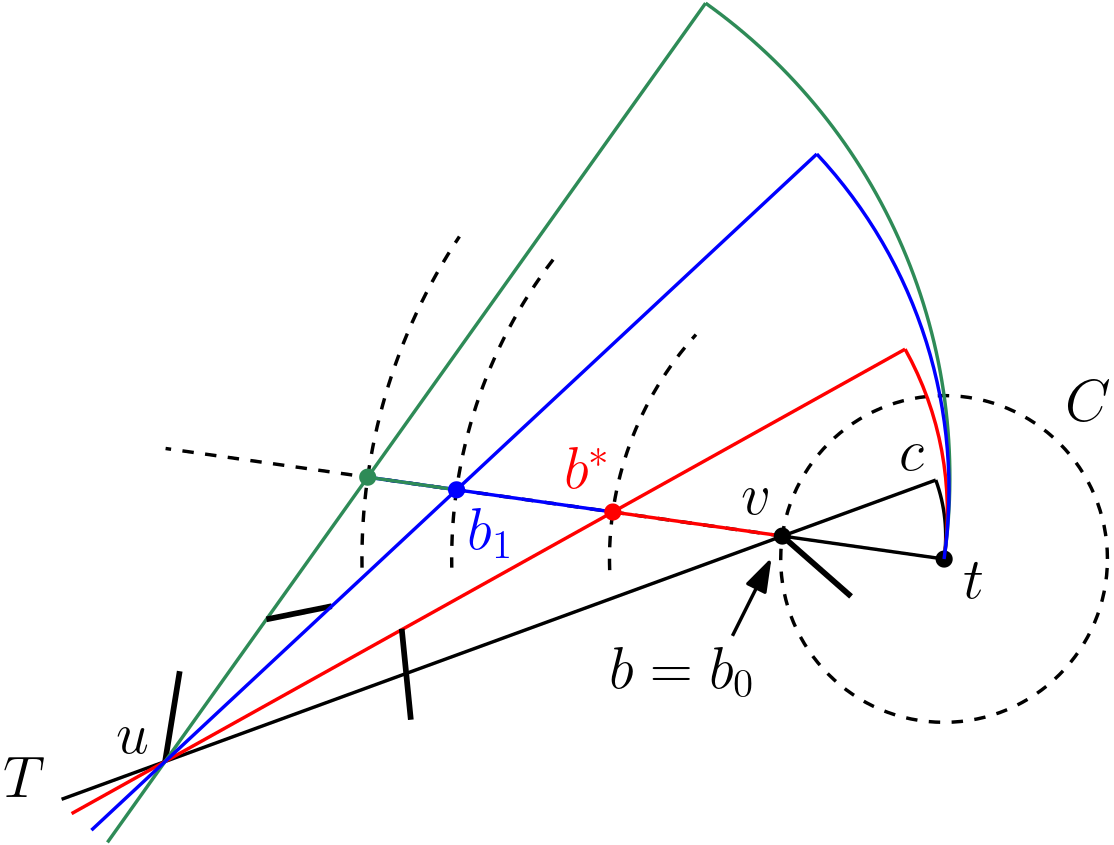}
	\caption{Finding an extremal feasible min-$r$ articulated probe trajectory in Scenario B.}
	\label{fig_min-r_2}
\end{figure}

In this case, the minimum value of $r$ for a possibly feasible trajectory occurs when $b = v$.
Let $b_0$ denote that location of $b$, and let $T$ be the corresponding trajectory.
If $T$ is free of obstacles, then $T$ is a feasible min-$r$ articulated trajectory (case IV of the lemma).

We now assume otherwise; that is, $T$ is an infeasible trajectory.
%Observe the following:	
%\renewcommand\labelenumi{\roman{enumi})}
%\renewcommand\theenumi\labelenumi
%\begin{enumerate}
Observe that if $\sigma_{bct}$ of $T$ intersects any obstacle line segment, then for certain there is no feasible articulated trajectory such that $ab$ of the trajectory passes through $u$ outside $C$, and $bt$ of the trajectory passes through $v$ inside $C$.
So, $\sigma_{bct}$ of $T$ must be empty of obstacle line segments.
Let $\rho_{b_0}$ denote the \emph{reversal} (i.e., opposite in direction) of the ray emanating from $b_0$ passing through $t$.
We increase $r$ by moving $b$ away from $b_0$ along $\rho_{b_0}$, while maintaining the intersection of $ab$ with $u$ and that of $bt$ with $v$, until the trajectory becomes feasible.
Observe that if $bt$, $bc$, or $\gamma_{ct}$ intersects an obstacle line segment at any given moment during the process of increasing $r$, then the trajectory would never become feasible thereafter. 
%\end{enumerate}

These observations imply that, when $b = b_0$, $ab$ of $T$ must be intersected by some obstacle line segment.
By increasing $r$, we may rid the trajectory of obstacle line segments that intersect $ab$.
Let $b^*$ denote the closest point to $b_0$ on $\rho_{b_0}$ for which the corresponding articulated trajectory is feasible.
Note that
%We can then conclude that, as we increase $r$, the trajectory becomes feasible (and $b = b^*$) when 
$ab$, when $b = b^*$, passes through an obstacle endpoint in addition to $u$ (case V of the lemma), as illustrated by the red trajectory in Figure \ref{fig_min-r_2}.

\begin{remark}
Observe that we can continue to increase $r$, while still having a feasible articulated trajectory, until $b$ reaches some point $b_1$, at which either i) $ab$, $bc$, or $\gamma_{ct}$ collides with an obstacle line segment, or ii) $\angle cbt = \pi/2$.
Let $r_{b_0}$, $r_{b^*}$, and $r_{b_1}$ be the lengths of $bc$ when $b = b_0$, $b = b^*$, and $b = b_1$, respectively, where $r_{b_0} \leq r_{b^*} \leq r_{b_1}$.
In addition, let $r_{\pi/2}$ be the length of $bc$ when $\angle cbt = \pi/2$.
According to our earlier arguments, $[r_{b^*}, r_{b_1}]$ is a feasible contiguous subset of $[r_{b_0}, r_{\pi/2}]$.
In fact, there could exist multiple (disjoint) contiguous feasible subsets of $[r_{b_0}, r_{\pi/2}]$, given that $ab$ may enter and leave intersections with multiple obstacle line segments during the process of increasing $r$, while $\sigma_{bct}$ remains free of obstacle line segments (refer to the blue and green trajectories in Figure \ref{fig_min-r_2} for an instance).
\end{remark}

This concludes the proof of Lemma \ref{lem1}.
\end{proof}

\begin{cor}
\label{cor}
Each value $r \in [r_{b^*}, r_{b_1}]$, as derived in the proof of Lemma \ref{lem1}, is associated with a feasible articulated trajectory.
\end{cor}

\begin{remark}
By computing every feasible contiguous range of $r$ in both Scenarios A and B for every pair of obstacle endpoints in $V$ (as described in the proof of Lemma \ref{lem1}), we can obtain all feasible values of $r$.
This observation would later become the basis of our solution approach for computing the entire feasible domain of $r$.
Briefly, the approach primarily consists of computing $[r_{b^*}, r_{b_1}]$ (as stated in Corollary \ref{cor}) for each pair of obstacle endpoints of $V$ in both Scenarios A and B.
However, recall that, for a given pair of obstacle endpoints $u, v \in V$ in Scenario B, $[r_{b^*}, r_{b_1}]$ may not be the only feasible contiguous range of $r$ -- that is, there may exist other disjoint feasible contiguous ranges of $r$ in addition to $[r_{b^*}, r_{b_1}]$ (see Figure \ref{fig_exp} for an example).
It turns out that, if handled with care, this situation can be addressed by using Scenario A
%, given the fact that the lower or upper bound of a feasible contiguous range of $r$, for the given pair of points $u, v$ in Scenario B, corresponds to $ab$ intersecting $u$ and some other obstacle endpoint $w \in V \setminus \{ v \}$ in Scenario A
(See Section \ref{sec:all_r} for details).
\end{remark}

\begin{figure}[h]
	\centering
	\includegraphics[scale=0.16]{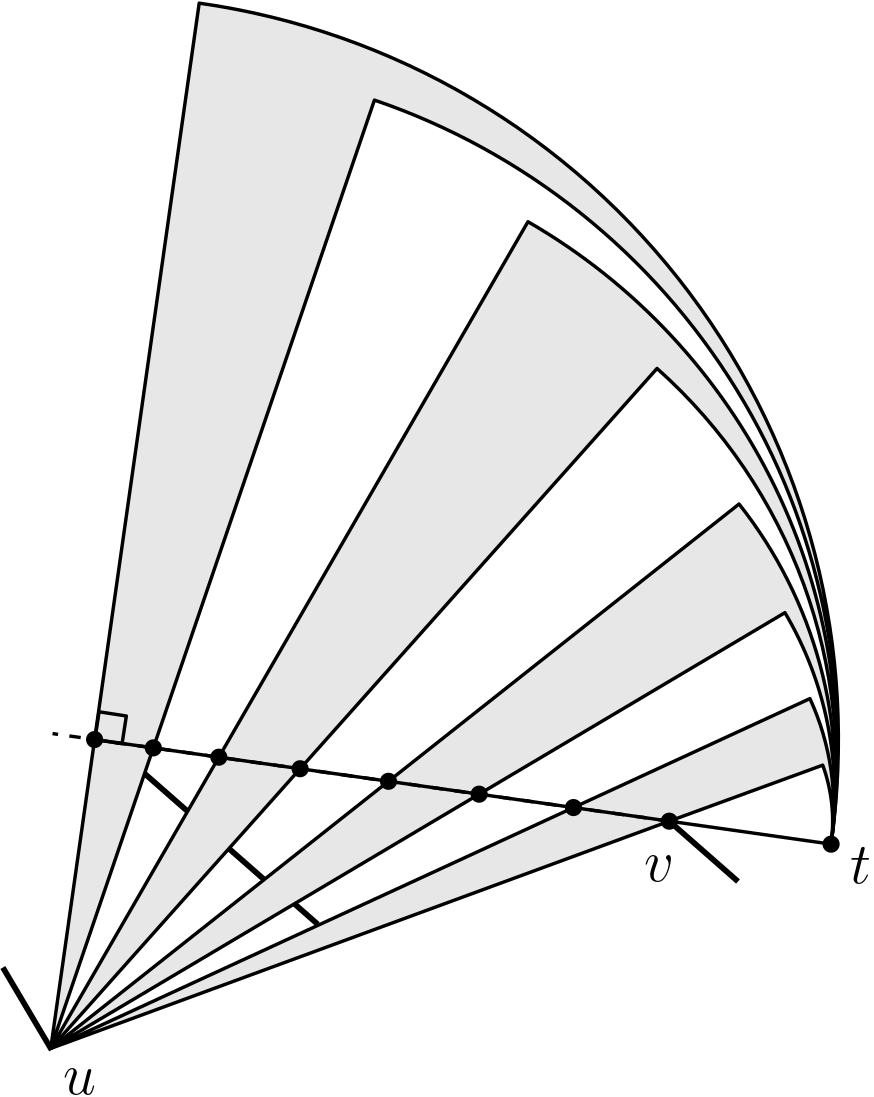}
	\caption{Multiple disjoint feasible ranges of $r$ (indicated by shaded regions) 
		%for an articulated probe trajectory, whose $ab$ and $bt$ intersect obstacle endpoints $u$ and $v$, respectively,
		in Scenario B.}
	\label{fig_exp}
\end{figure}

\subsection{A simple solution approach -- Enumerate and verify}
\label{enum_and_verify}
Based on Lemma \ref{lem1}, we can devise a simple approach to solve Problem \ref{prob_min_r} as follows.
For each pair of obstacle endpoints $u, v \in V$, find an extremal feasible min-$r$ trajectory, if one exists, as described in Lemma \ref{lem1}, and report the feasible trajectory with the smallest value of $r$ overall.
%Among those that are feasible, the extremal trajectory with the minimum $r$ is a feasible min-$r$ articulated trajectory.

The set of extremal min-$r$ articulated trajectories is characterized by $O(n^3)$ combinatorial events.
Thus, we can enumerate these extremal trajectories in $O(n^3)$ time by using an algebraic-geometric approach.
An extremal min-$r$ articulated trajectory is considered feasible if and only if both i) $ab$ and ii) $\sigma_{bct}$ do not intersect any obstacle line segment.
Checking for these scenarios can be reduced to the following two query problems -- i) ray shooting query, which can be answered in $O(\log n)$ time after a preprocessing that takes $O(n^2)$ time and space \cite{pocchiola90graphics}, and ii) circular sector emptiness query.
%\subparagraph{Ray shooting queries.}  As described by Pocchiola \cite{pocchiola90graphics}, a set $P$ of $n$ disjoint line segments can be preprocessed, in $O(n^2)$ time, into a data structure of size $O(n^2)$ so that one can determine if a query ray intersects $P$ in $O(\log n)$ time.
%\subparagraph{Circular sector emptiness queries.}
A solution to the latter problem
%(designated as Subproblem \ref{subprob1})
calls for the construction of lower envelopes of $n$ bivariate partial functions in three dimensions
(see Subproblem \ref{subprob1} in Section \ref{query_A} for details).
The query data structure is of size $O(n^{2+\epsilon})$ and can be constructed in $O(n^{2+\epsilon})$ time.
Using the data structure, a query can be answered in $O(\log n)$ time.
Since we have $O(n^3)$ queries in the worst case, the following result is obtained.

\begin{thm}
	\label{thm2}
	A feasible min-$r$ articulated probe trajectory, if one exists, can be determined in $O(n^3 \log n)$ time using $O(n^{2+\epsilon})$ space, for any constant $ \epsilon > 0$.
\end{thm}

As an alternative, we can perform $O(n)$ intersection/emptiness checks on each of the $O(n^3)$ extremal trajectories, and obtain an $O(n^4)$-time algorithm with an $O(n)$ space usage.

\subsection{An improved solution approach}
\label{improved}

Obviously, if we rely on identifying critical events (i.e., $O(n^3)$ combinatorial possibilities) as the basis for our solution approach, $O(n^3 \log n)$ would likely be the best attainable running time.
That raises the question of whether we could do better if we, at least partially, forego the notion of finding extremal trajectories.

We begin by emphasizing that, as stated in Lemma \ref{lem1}, an extremal feasible min-$r$ articulated trajectory passes through two obstacle endpoints, neither of which is inside $C$.
Consider the following solution approach.
For each point $v \in V$, compute the set $R_v$ of rays with the following properties:
i) Each ray originates at $v$ and passes through a point $u \in V \setminus \{v\}$.
ii) Line segment $vb_0$ does not intersect any line segment of $P$, where $b_0$ is the point of tangency between the supporting line of the ray and the circle $C$ centered at $t$ (Figure \ref{fig_steps_A}A).
iii) If the ray passes through $b_0$, then the \emph{reversal} of the ray does not intersect any line segment of $P$;
otherwise, the ray itself does not intersect any line segment of $P$.
Set $R_v$ can be obtained in $O(n \log n)$ time by computing the visibility polygon from $v$ \cite{arkin87opt,heffernan95opt,suri86worst}.
Since $|V| = O(n)$, the worst-case running time for finding the set of rays $R = \cup_{v \in V} R_v$ is $O(n^2 \log n)$.

\begin{figure}
	\centering
	\includegraphics[scale=0.16]{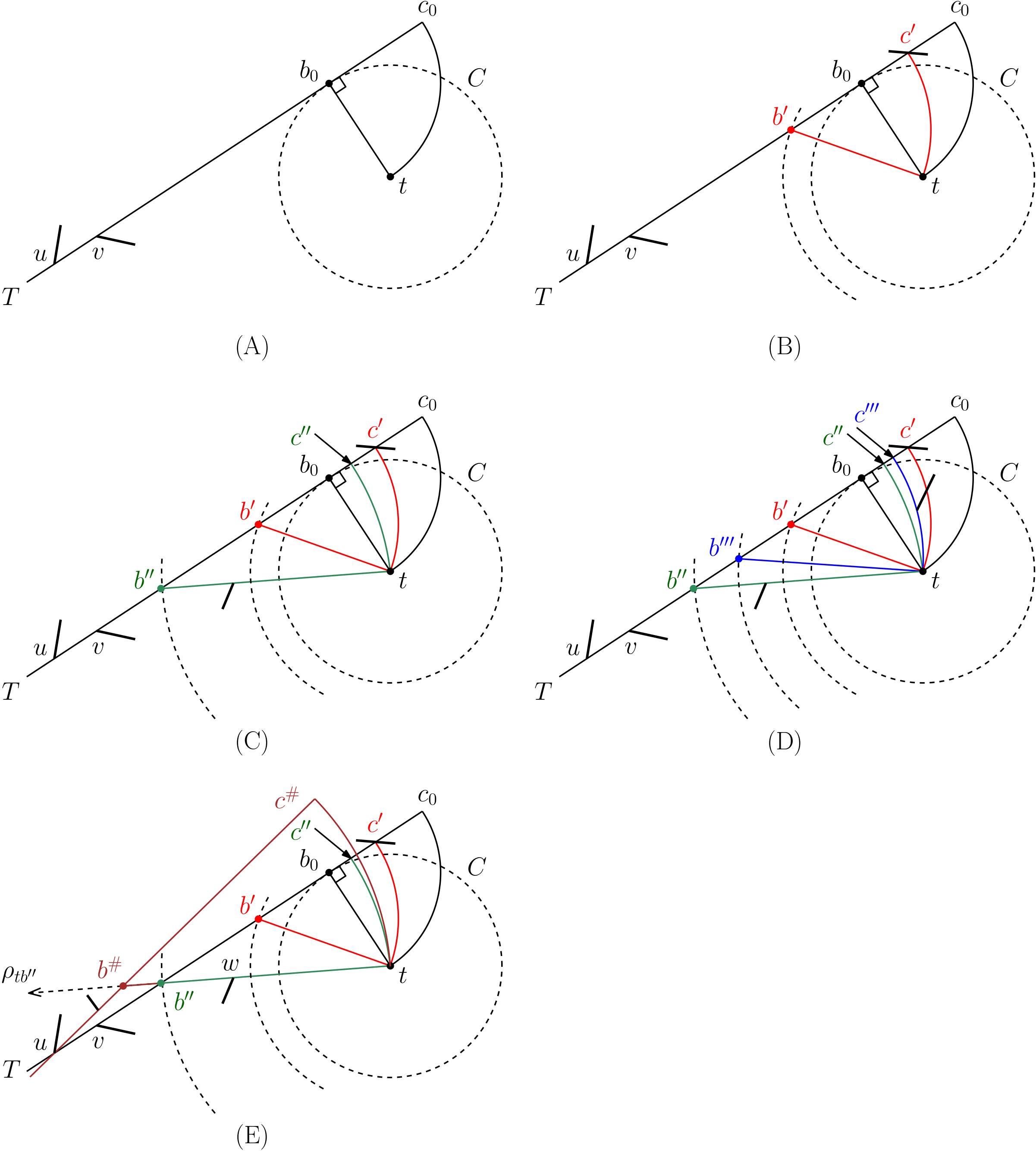}
	\caption{Illustrations of steps 
		%the sequence of steps to find an extremal feasible min-$r$ articulated trajectory.
		(A) A1 and A2,
		(B) A3 and A4,
		(C) A5 and A6, and
		(D) A7.
		(E) A6s.
	}
	\label{fig_steps_A}
\end{figure}

Note that each ray of $R$ is associated with a trajectory $T$ that has an obstacle-free line segment $ab$ passing through two obstacle endpoints.
Without loss of generality, assume that $ab$ of $T$ passes through a pair of obstacle endpoints $u, v \in V$, where $u \neq v$, in the way depicted in Figure \ref{fig_steps_A}A.
Assume that $bc$ of $T$ rotates clockwise to reach $t$ (the other case is symmetrical).
%Recall that $\angle cbt$ denotes the angle of rotation of segment $bc$ of the probe in order to reach $t$.
Let $b_0$ be the position of $b$ when $\angle cbt = \pi/2$ radians, and $c_0$ be the position of $c$ when $b = b_0$.
In order to find a feasible min-$r$ articulated trajectory, we perform the following sequence of steps.

\begin{enumerate}
	\renewcommand\labelenumi{A\arabic{enumi}.}
	\renewcommand\theenumi\labelenumi
	
	\item Check if the articulated trajectory $T$ with $\angle cbt = \pi/2$ radians is feasible.
	Specifically, check if the quarter circular sector bounded by $b_0c_0$, $b_0t$, and circular arc $\gamma_{c_0t}$ (centered at $b_0$ and emanating counter-clockwise from $t$ to $c_0$) is free of obstacles (Figure \ref{fig_steps_A}A).
	If it is, then $T$ is a feasible min-$r$ articulated trajectory whose $ab$ passes through $u$ and $v$.
	Otherwise, proceed with step A2.
	
	\item Check if $b_0t$ is intersected by any obstacle (Figure \ref{fig_steps_A}A).
	If it is, then a feasible min-$r$ articulated trajectory whose $ab$ passes through $u$ and $v$ does not exist.
	Otherwise, proceed with steps A3 and A4.
	
	\item Find the closest point $c' \in b_0c_0$ to $c_0$ such that $b_0c'$ does not intersect any obstacle (Figure \ref{fig_steps_A}B).
	Compute the center $b'$ of the circular arc $\gamma_{c't}$ emanating counter-clockwise from $t$ to $c'$, where $b' \in vb_0$.
	
	\item Check if $b't$ is intersected by any obstacle (Figure \ref{fig_steps_A}B).
	If it is, then a feasible min-$r$ articulated trajectory whose $ab$ passes through $u$ and $v$ does not exist.
	Otherwise, proceed with steps A5 and A6.
	
	\item Find the closest point $b'' \in vb'$ to $b'$ such that $b''t$ passes through an obstacle endpoint (Figure \ref{fig_steps_A}C).
	Compute the corresponding point $c''$ (i.e., the intersection between $b_0c'$ and the circle of radius $|b''t|$ centered at $b''$).
	Note that the triangle bounded by $b't$, $b''t$, and $b'b''$ is free of obstacles.
	
	\item Check if the ``sector'' bounded by $b'c''$, $b't$, and circular arc $\gamma_{c''t}$ (centered at $b''$ and emanating counter-clockwise from $t$ to $c''$) intersects any obstacle (Figure \ref{fig_steps_A}C).
	Note that this is equivalent to checking if the circular sector bounded by $b''c''$, $b''t$, and $\gamma_{c''t}$ intersects any obstacle.
	If it does, then a feasible min-$r$ articulated trajectory whose $ab$ passes through $u$ and $v$ does not exist.
	Otherwise, proceed with steps A6s and A7.
	
	\renewcommand\labelenumi{A\arabic{enumi}s.}
	\renewcommand\theenumi\labelenumi
	\setcounter{enumi}{5}
	
	\item
	\textbf{This side step after A6 is only necessary when computing the feasible domain of \boldmath$r$ (see Section \ref{sec:all_r} for details).}
	Let $\rho_{tb''}$ be the ray originating at $t$ and passing through $b''$ (Figure \ref{fig_steps_A}E).
	Let $w \in \rho_{tb''}$ denote the obstacle endpoint intersected by $b''t$, and
	$b_{\pi/2}$ be the point on $\rho_{tb''}$ such that $ub_{\pi/2}$ forms a right angle with $b_{\pi/2}b''$.
	While maintaining the intersection of $ab$ (of the probe) with $u$, and that of $bt$ with $w$, find the closest point $b^\# \in b_{\pi/2}b''$ to $b''$ such that
	i) either $a^\#u$ or $ub^\#$ passes through an obstacle endpoint (other than $v$),
	ii) circular arc $\gamma_{c^\#t}$ (centered at $b^\#$ and emanating counter-clockwise from $t$ to $c^\#$) intersects an obstacle endpoint or is tangent to an obstacle line segment, 
	iii) $b^\#c^\#$ passes through an obstacle endpoint, or 
	iv) $b^\#c^\#$ intersects an obstacle line segment at $c^\#$, 
	where $a^\#$ and $c^\#$ are the positions of $a$ and $c$ (of the probe), respectively, when $b = b^\#$.
	
	\renewcommand\labelenumi{A\arabic{enumi}.}
	\renewcommand\theenumi\labelenumi
	\setcounter{enumi}{6}
	
	\item At this point, observe that the articulated trajectory with the intermediate configuration represented by $ab''c''$ is feasible.
	Find the closest point $b''' \in b'b''$ to $b'$ such that circular arc $\gamma_{c'''t}$ (centered at $b'''$ and emanating counter-clockwise from $t$ to $c'''$) intersects an obstacle endpoint or is tangent to an obstacle line segment (Figure \ref{fig_steps_A}D).
	Note that the ``sector'' bounded by $b'c'''$, $b't$, and circular arc $\gamma_{c'''t}$ is free of obstacles.
	The articulated trajectory with the intermediate configuration indicated by $ab'''c'''$ is a feasible min-$r$ articulated trajectory whose $ab$ passes through $u$ and $v$.
\end{enumerate}

\renewcommand{\arraystretch}{1.5}
\begin{table}
	\footnotesize
	\centering
	\caption{Summary of query data structures used in steps A1-A7.
		The size, preprocessing time, and query time of a data structure are denoted by $S(n)$, $P(n)$, and $Q(n)$, respectively.
		Our results are highlighted in gray.}
	\vspace*{3mm}
	\label{tab1}
	\begin{tabular}{ l >{\raggedright\arraybackslash} p{4.8cm} l l l }
		\hline
		\textbf{Step} & \textbf{Query} & \boldmath$S(n)$ & \boldmath$P(n)$ & \boldmath$Q(n)$ \\
		\hline
		
		\rowcolor{Gray}
		A1, A6 & 
		Subproblem \ref{subprob1}: 
		\newline Circular sector emptiness queries &
		$O(n^{2+\epsilon})$ & $O(n^{2+\epsilon})$ & $O(\log n)$ \\ 
		
		\rowcolor{Gray}
		A2, A4 & 
		Subproblem \ref{subprob2}: 
		\newline Radius intersection queries &
		$O(n)$ & $O(n \log n)$ & $O(\log n)$ \\
		
		A3 & Ray shooting queries \cite{pocchiola90graphics} &
		$O(n^2)$ & $O(n^2)$ & $O(\log n)$ \\
		
		\rowcolor{Gray}
		A5 & 
		Subproblem \ref{subprob3}:
		\newline Radius shooting queries &
		$O(n^2 / \log^2 n)$ & $O(n^2 / \log^2 n)$ & $O(\log^2 n)$ \\
		
		\rowcolor{Gray}
		A7 & 
		Subproblem \ref{subprob4}:
		\newline Arc shooting queries &
		$O(n^{2+\epsilon})$ & $O(n^{2+\epsilon})$ & $O(\log n)$ \\ 
		
		\hline
	\end{tabular}
\end{table}
\normalsize

By simply performing an $O(n)$-check (i.e., check against each of the $O(n)$ obstacles) in each of the steps above, we can obtain an $O(n^3)$-time ``brute-force'' method to find a feasible min-$r$ articulated trajectory, if one exists.
Alternatively, we can address these steps using efficient data structures, which require geometric constructs such as lower envelopes and half-space decomposition schemes.
Refer to Table \ref{tab1} for a summary of the query data structures, whose details are presented next in Section \ref{query_A}.
$O(n^2)$ queries are to be processed in the worst case, resulting in a total query time bounded by $O(n^2 \log^2 n)$.
Since the preprocessing time of the query data structures is dominant overall, we have the following final result.

\begin{thm}
	A feasible min-$r$ articulated probe trajectory, if one exists, can be determined in $O(n^{2+\epsilon})$ time using $O(n^{2+\epsilon})$ space, for any constant $\epsilon > 0$.
\end{thm}

\subsection{Query problems in steps A1-A7}
\label{query_A}

\paragraph{Steps A1 and A6 (circular sector emptiness queries)}
Consider the following circular sector emptiness query problem in two dimensions.

\begin{subprob}
	\label{subprob1}
	Given a set $P$ of $n$ line segments and a fixed point $t$, preprocess them so that, for a query circular sector $\sigma$ with an endpoint of its arc located at $t$, one can efficiently determine whether $\sigma$ intersects $P$.
\end{subprob}

Note that Subproblem \ref{subprob1} is slightly different from the circular sector emptiness query problem addressed in \cite{teo20traj}, given that a query circular sector may not have a fixed radius $r$.

Let $C$ denote the circle of radius $r$ centered at $t$.
For any point $b \in C$, let $\theta$ be the angle of $tb$ relative to the $x$-axis.
Let $D$ be the circle of radius $r$ centered at $b$.
Circle $D$ is uniquely defined by $\theta$ and $r$ as $b$ lies on $C$.
Let $\theta \in [0, 2\pi)$ and $r \in (0, R]$, where $R$ is the fixed radius of workspace $S$.
Define a partial function $\rho_s : [0, 2\pi) \times (0, R] \rightarrow \Re_{\geq 0}$ as follows.
Let $bc'$ be the farthest radius of $D$ counter-clockwise from $bt$ before the minor circular sector bounded by radii $bt$ and $bc'$ intersects a given line segment $s$ of $P$.
Let $\rho_s$ be the angle of $bc'$ measured counter-clockwise from $bt$ (the clockwise case can be handled symmetrically).
Since line segment $bc$ of the probe, after the initial insertion of the probe has completed, may only rotate up to $\pi/2$ radians in either direction, $\rho_s \in [0, \pi/2]$.

By following the same rationale and argument in \cite[Section 3]{daescu19traj3d}, we can claim that $\rho_s(\theta, r)$ is an inverse trigonometric function, and we can define an algebraic function $f_s = \sin (\rho_s/2)$ in terms of three variables $x_b$, $y_b$, and $r$, where $x_b$ and $y_b$ are the $x$- and $y$-coordinates of $b \in C$.
Since $\rho_s$ is partially defined and continuous over $\theta$ and $r$, so is $f_s$ over $x_b$, $y_b$, and $r$.

Since ${x_b}^2 + {y_b}^2 = r^2$, we can construct two lower envelopes $V_1$ and $V_2$ of the piecewise algebraic functions $f_s$ for all given line segments $s$ of $P$, such that $V_1$ is the lower envelope of $f_s$ for $y_b \geq 0$, and $V_2$ is for $y_b < 0$.
Consequently, $V_1$ and $V_2$ are functions of only two variables $x_b$ and $r$.

By using the deterministic divide-and-conquer algorithm given by Agarwal et al. \cite{agarwal96overlay}, we can construct the lower envelopes of $n$ bivariate piecewise algebraic functions (of constant degree), namely $V_1$ and $V_2$, in $O(n^{2+\epsilon})$ time using $O(n^{2+\epsilon})$ space, for any $\epsilon > 0$.

Given a query circular sector $\sigma$, let $b_\sigma$ denote the apex of $\sigma$, and $r_\sigma$ be the radius of $\sigma$.
If $y_{b_\sigma} \geq 0$, then $(b_\sigma, r_\sigma)$ is looked up in $V_1$ in $O(\log n)$ time (with the aid of a supporting point location query data structure); otherwise, $(b_\sigma, r_\sigma)$ is looked up in $V_2$.
Let $\rho_\sigma$ be the acute angle between the two bounding radii of $\sigma$.
If $\sin (\rho_\sigma/2)$ is less than $f_s(x_{b_\sigma}, r_\sigma)$ for all $s \in P$, then $\sigma$ does not intersect $P$.

\begin{lem}
	For a fixed point $t$ and any constant $\epsilon > 0$, a set $P$ of $n$ line segments can be preprocessed into a data structure of size $O(n^{2+\epsilon})$ in $O(n^{2+\epsilon})$ time so that, for a query circular sector $\sigma$ with an endpoint of its arc located at $t$, one can determine whether $\sigma$ intersects $P$ in $O(\log n)$ time.
\end{lem}

\paragraph{Steps A2 and A4 (radius intersection queries)}
Steps A2 and A4 can be essentially reduced to the following radius intersection query problem.

\begin{subprob}
	\label{subprob2}
	Given a set $P$ of $n$ line segments and a fixed point $t$, preprocess them so that, for a query radius $r$ of a circle centered at $t$, one can efficiently determine whether $r$ intersects $P$.
\end{subprob}

Note that the length of $r$ is given at query time.
For a slight abuse of notation, $r$ is used to denote a query line segment (radius) as well as its length.

As with the circular arc intersection query problem in \cite{teo20traj}, we can construct a lower envelope in two dimensions as follows.
For any point $p \in \Re^2$, let $\theta$ denote the angle of $tp$ measured counter-clockwise from the $x$-axis, and $\theta \in [0, 2\pi)$.
For each line segment $s \in P$, we define $f_s(\theta)$ to be the length of the ray from $t$ to its intersection with $s$.
Observe that $f_s(\theta)$ is a partially defined function over at most two maximal contiguous subsets of $[0, 2\pi)$.
In addition, given two disjoint line segments $s_i$ and $s_j$ (which could possibly intersect at their endpoints), $f_{s_i}(\theta) = f_{s_j}(\theta)$ occurs for at most one value of $\theta$, which corresponds to the common intersection point.

Let $\mathcal{V}$ be the lower envelope of $f_s(\theta)$ for all given line segments $s \in P$.
Given the properties of each $f_s(\theta)$, the size of $\mathcal{V}$ is bounded by the second-order Davenport-Schinzel sequence, which is $O(n)$ in complexity, and we can compute $\mathcal{V}$ in $O(n \log n)$ time \cite{Hersh89find,sharir95dav}.

In order to determine if a query radius of length $r$ intersects $P$, the angle $\theta_r$ of the query radius is looked up in $\mathcal{V}$ by using a binary search that takes $O(\log n)$ time.
If $r$ is less than $f_s(\theta_r)$ for all $s$ of $P$, then the query radius does not intersect any line segment of $P$.

\begin{lem}
	A set $P$ of $n$ line segments and a fixed point $t$ can be preprocessed into a data structure of size $O(n)$ in $O(n \log n)$ time so that, for a query radius $r$ of a circle centered at $t$, one can efficiently determine whether $r$ intersects $P$ in $O(\log n)$ time.
\end{lem}

\paragraph{Step A3 (ray shooting queries)}
Refer to \cite{pocchiola90graphics}.

\paragraph{Step A5 (radius shooting queries)}
The query problem associated with step A5 can be described as follows.

\begin{subprob}
	\label{subprob3}
	Given a set $P$ of $n$ points in the plane and a fixed origin $t$, preprocess them so that, for a query line $L$ and a query angle $\alpha$ (relative to the $x$-axis), report the point in $P$ that lies on the side of $L$ containing the origin $t$ and has the smallest angle (from $t$ relative to the $x$-axis) greater than $\alpha$.
\end{subprob}

We consider the half-space decomposition scheme as described by Matou{\v{s}}ek in \cite[Theorem 5.1]{matouvsek93range}.
Given a set $P$ of $n$ points in the plane and a parameter $r \geq n$, one can built a data structure $\mathcal{D}$ with the following properties.
Data structure $\mathcal{D}$ contains a collection of canonical subsets of $P$.
These canonical sets are divided into two groups -- i) inner subsets and ii) remainder subsets.
Inner subsets can be partitioned into $O(\log r)$ collections $C_0, \dots, C_{k-1}$ such that each collection $C_i$ contains $O(\rho^i)$ inner subsets of size at most $n/\rho^i$ each, where $\rho$ is a constant and $\rho^{k-1} \leq r \leq \rho^k$.
There are $O(r^2)$ remainder subsets, each of which is at most $n/r$ in size.
Data structure $\mathcal{D}$ can be computed in $O(nr)$ time.

Using data structure $\mathcal{D}$, given a half-plane $H$, one can, in $O(\log r)$ time, find i) one inner subset from each collection $C_i$ that is contained in $P \cap H$, and ii) a single remainder subset that contains the points not covered by any of the inner subsets (but may contain points other than $P \cap H$).

We now describe a data structure for solving Subproblem \ref{subprob3}.
Set $r = n / \log^2 n$ and compute the decomposition scheme aforementioned as follows.
Each remainder subset is of size $O(\log^2 n)$ and is stored explicitly.
For each inner subset, we store the $O(n/\rho^i)$ points in the sorted order according to their angles from the origin $t$ (relative to the $x$-axis), so that a binary search query can be performed in $O(\log n)$ time to find the point in the inner subset that has the smallest angle greater than $\alpha$.
The total size of the data structure is $O(n^2 / \log^2 n)$, and it can be built in $O(n^2 / \log^2 n)$ time.

Given a query line $L$ and a query angle $\alpha$, let $H$ be the half-plane that is delimited by $L$ and containing the origin $t$.
The decomposition scheme is then queried as follows.
In order to find the point in $P \cap H$ that has the smallest angle greater than $\alpha$, check every point in the one remainder subset (of $O(\log^2 n)$ in size) intersecting $H$, and perform a binary search query in each of the $O(\log n)$ inner subsets contained in $P \cap H$.
The total query time is therefore bounded by $O(\log^2 n)$.

\begin{lem}
	For a fixed origin $t$, a set $P$ of $n$ points in the plane can be preprocessed into a data structure of size $O(n^2 / \log^2 n)$ in $O(n^2 / \log^2 n)$ time so that, for a query line $L$ and a query angle $\alpha$ (relative to the $x$-axis), one can report the point in $P$ that lies on the side of $L$ containing the origin $t$ and has the smallest angle (from $t$ relative to the $x$-axis) greater than $\alpha$ in $O(\log^2 n)$ time.
\end{lem}

\begin{remark}
In general, a trade-off between space and time usage by the half-space decomposition scheme can be achieved by following the strategy outlined in \cite[Theorem 6.2]{matouvsek93range}.
\end{remark}

\paragraph{Step A7 (arc shooting queries)}
A formal statement of the query problem involved in step A7 is given below.

\begin{subprob}
	\label{subprob4}
	Given a fixed origin $t$, for a query line $L$, let $bt$ be the perpendicular line segment dropped from the origin $t$ to line $L$, where $b \in L$ (Figure \ref{fig_subprob4_1}).
	Assume that $L$ is located above $t$ (the other case can be described similarly due to its symmetry), and $b$ partitions $L$ into a left half-line $L^-$ and a right half-line $L^+$.
	Let $b^*$ be a point on $L^-$.
	Let $\gamma_{c^*t}$ denote the circular arc that is centered at $b^*$ and emanating counter-clockwise from $t$ to $c^*$, where $c^* \in L^+$.
	Given a set $P$ of $n$ line segments, preprocess $P$ so that one can efficiently find the closest point $b^*$ to $b$ on $L^-$ such that the ``sector'' bounded by $bc^*$, $bt$, and $\gamma_{c^*t}$ is free of $P$, if such $b^*$ exists at all.
\end{subprob}

\begin{figure}[h]
	\centering
	\includegraphics[scale=0.17]{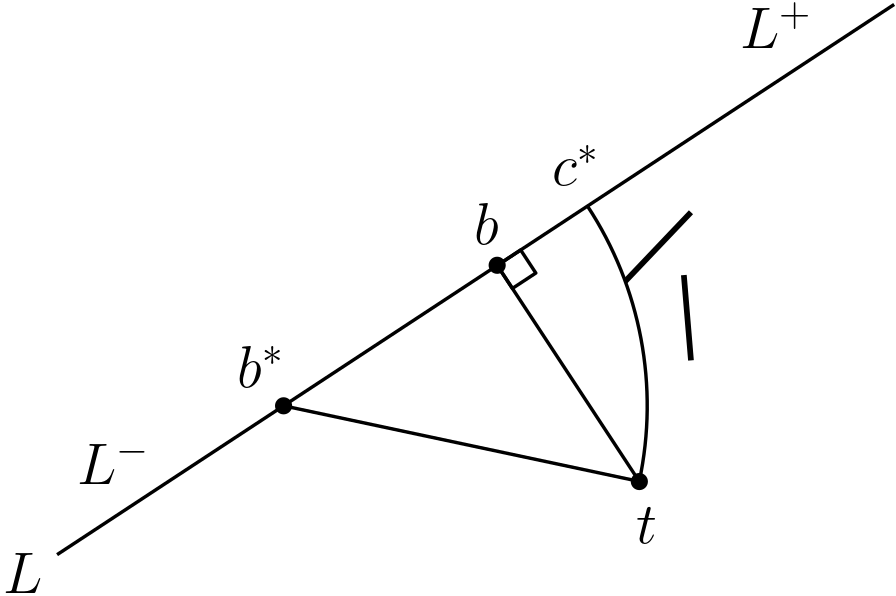}
	\caption{Illustration of the closest point $b^*$ to $b$ on $L^-$ such that the region bounded by $bc^*$, $bt$, and $\gamma_{c^*t}$ is free of line segments.}
	\label{fig_subprob4_1}
\end{figure}

Given a query line $L$, let $r$ denote the length of its corresponding segment $bt$, and $\theta$ be the angle of $bt$ measured counter-clockwise from the $x$-axis (Figure \ref{fig_subprob4_2}).
Note that $L$ is uniquely characterized by $(r, \theta)$.
In addition, $L$ can be expressed as $y = mx + d$, where $m$ is the slope of $L$, and $d$ is the $y$-intercept of $L$. Observe that $m = -1/\tan \theta$ and $d = r/\sin \theta$, where $\theta \in [0, 2\pi)$, $r \in [0, R]$, and $R$ is the constant radius of the circular workspace $S$.

\begin{figure}[h]
	\centering
	\includegraphics[scale=0.17]{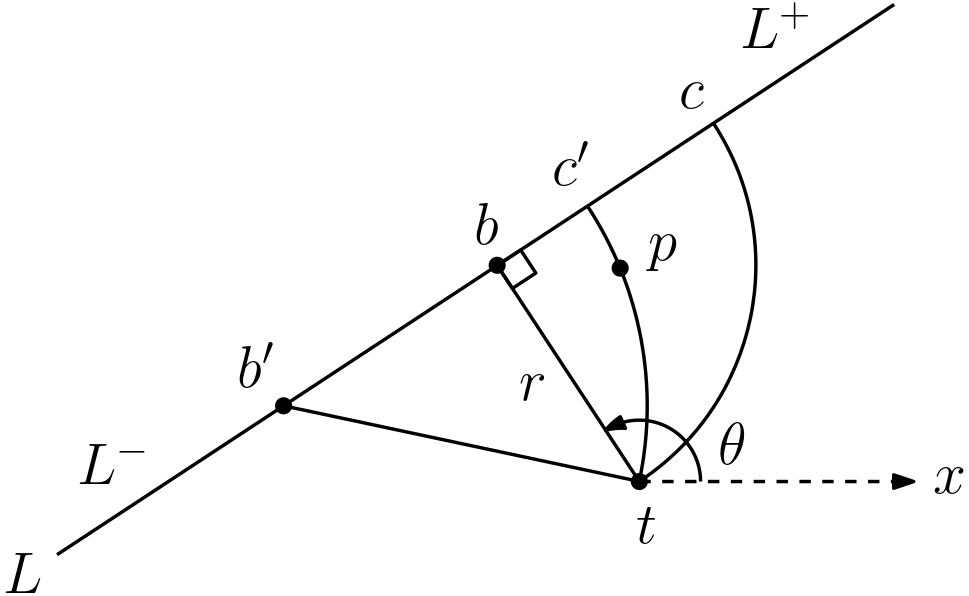}
	\caption{Notations used in the exposition of Subproblem \ref{subprob4}.}
	\label{fig_subprob4_2}
\end{figure}

For simplicity of exposition, let $P$ be a set of $n$ obstacle endpoints instead of line segments (one would later find that the same argument holds for line segments as well).

Let $\gamma_{ct}$ denote the circular arc that is centered at $b$ and emanating counter-clockwise from $t$ to $c$, where $c \in L^+$.
Let $\sigma_{bct}$ be the quarter circular sector bounded by $bc$, $bt$, and $\gamma_{ct}$. 
As with $L$, $\sigma_{bct}$ is uniquely determined by $r$ and $\theta$ (i.e., $m$ and $d$).
For a point $p \in P$, there exists a circle that passes through $t$ and $p$ and is centered at a point $b'$ on $L^-$.
The $x$- and $y$-coordinates of point $b'$ can be written, respectively, as
\begin{equation}
\label{eqn1}
x_{b'} = \frac{\frac{{x_p}^2 + {y_p}^2}{2y_p} - d}{m + \frac{x_p}{y_p}} \quad \text{and} \quad
y_{b'} = \frac{m \left( \frac{{x_p}^2 + {y_p}^2}{2y_p} - d \right)}{m + \frac{x_p}{y_p}} + d
\end{equation}
where $x_p$ and $y_p$ denote the $x$- and $y$-coordinates of $p$, respectively.
On the other hand, the $x$- and $y$-coordinates of point $b$ are given, respectively, by
\begin{equation}
\label{eqn2}
x_b = \frac{d}{m + \frac{1}{m}} \quad \text{and} \quad
y_b = \frac{d}{m \left( m + \frac{1}{m} \right)}
\end{equation}
Let $\eta_p$ denote the distance from $b$ to $b'$.
Distance $\eta_p$ can be expressed as
\begin{align}
\label{eqn3}
\eta_p &= \|b-b'\| \nonumber \\ 
&= \sqrt{\|b\|^2 + \|b'\|^2 - 2b \cdot b'} \nonumber \\
&= \sqrt{{x_b}^2 + {y_b}^2 + {x_{b'}}^2 + {y_{b'}}^2 - 2 \left( x_bx_{b'} + y_by_{b'} \right)}
\end{align}
where $\| \|$ symbolizes the Euclidean norm, and $\cdot$ denotes the dot product.
By substituting Equations \ref{eqn1} and \ref{eqn2} into Equation \ref{eqn3}, $\eta_p$ can be expressed in terms of $m$ and $d$.
Notice that circular arc $\gamma_{c't}$, which is centered at $b'$ and emanating counter-clockwise from $t$ to $c'$, is always contained within the quarter circular sector $\sigma_{bct}$.
Thus, we are only concerned with characterizing $\eta_p$ for points $p \in P$ that are located within the quarter circular sector $\sigma_{bct}$.
The result is that, for each point $p \in P$, $\eta_p$ is only defined over a partial contiguous region (subset) $F_p$ of the $(m, d)$-space (see Figure \ref{fig_subprob4_3} for an illustration).

\begin{figure}[h]
	\centering
	\includegraphics[scale=0.2]{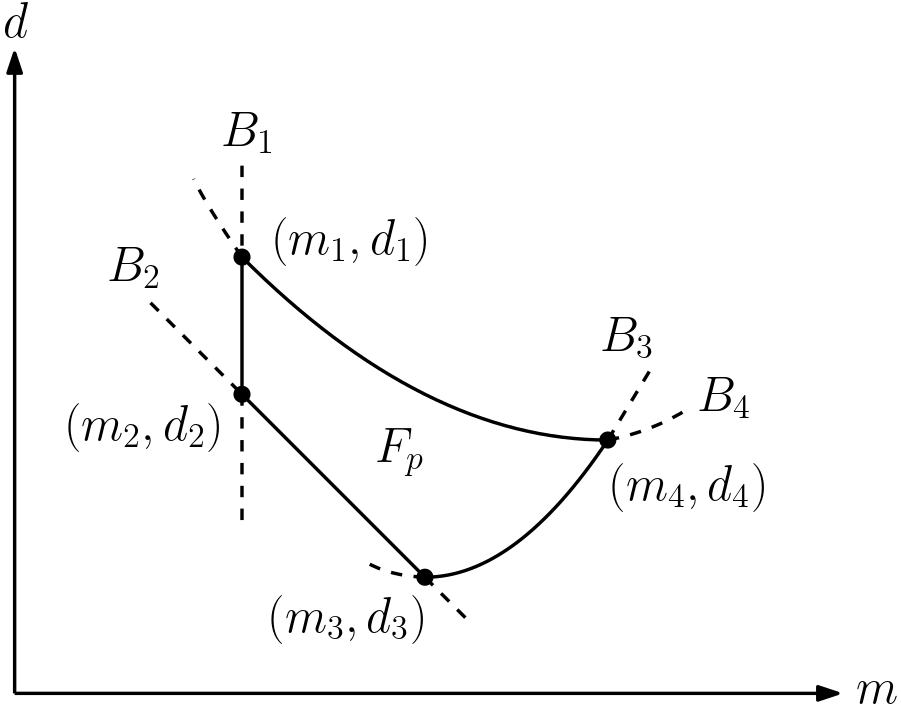}
	\caption{Contiguous region $F_p$ for a point $p$ is bounded by four polynomial curves in the $(m, d)$-space.}
	\label{fig_subprob4_3}
\end{figure}

In fact, for a point $p \in P$, $F_p$ is the intersection of four polynomial inequalities of constant degree (dependent on $p$), which can be determined using algebraic geometry.
The four bounding polynomial functions of $F_p$ are

\begin{enumerate}
	\item \( \displaystyle m = - \frac{x_p}{y_p} \),
	
	\item \( \displaystyle d = y_p - mx_p \),
	
	\item \( \displaystyle d = y_0 - mx_0 \) where \\
	\( \displaystyle {x_0}^2 + {y_0}^2 = r^2 \) and
	\( \displaystyle \left( 1 + \frac{1}{m^2} \right) \left( \frac{d}{m + \frac{1}{m}} \right)^2 = r^2 \), and
	
	\item \( \displaystyle \left( 1 + \frac{1}{m^2} \right) \left( \frac{d}{m + \frac{1}{m}} \right)^2 = R^2 \), \\
\end{enumerate}
which correspond to curves $B_1$, $B_2$, $B_3$, and $B_4$, respectively, in Figure \ref{fig_subprob4_3}.
Intuitively, traveling along the boundary of $F_p$ (i.e., $(m_1, d_1) \rightarrow (m_2, d_2) \rightarrow (m_3, d_3) \rightarrow (m_4, d_4) \rightarrow (m_1, d_1)$ in Figure \ref{fig_subprob4_3}) is analogous to changing (the slope $m$ and $y$-intercept $d$ of) query line $L$ from one possible extreme to another while keeping $p$ within or on the boundary of $\sigma_{bct}$ (i.e., $L_1 \rightarrow L_2 \rightarrow L_3 \rightarrow L_4 \rightarrow L_1$ in Figure \ref{fig_subprob4_4}).

\begin{figure}[h]
	\centering
	\includegraphics[scale=0.2]{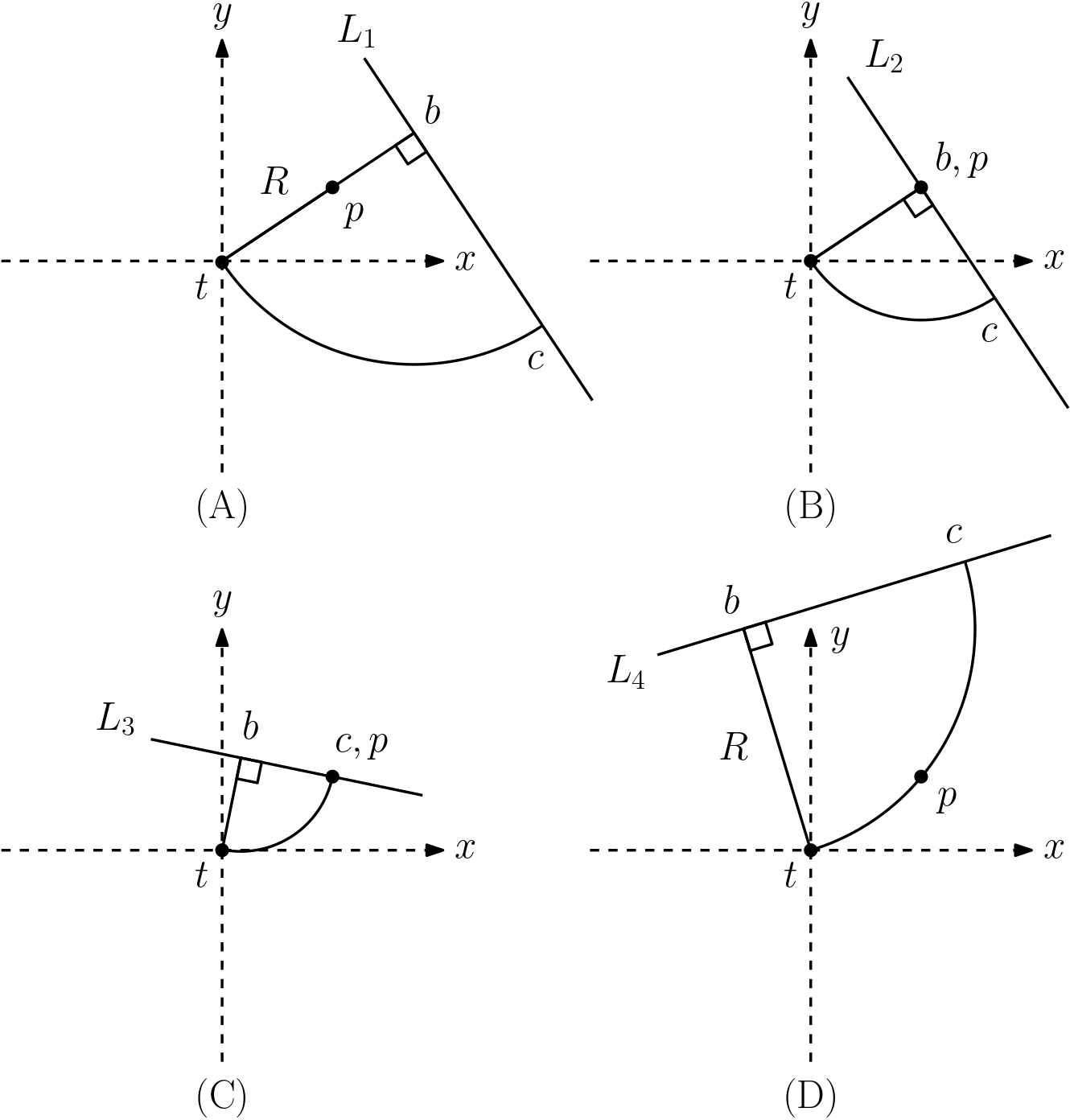}
	\caption{Illustrations of the boundary conditions in the $(x, y)$-plane, indicated by the changes of query line $L$ (and its associated quarter circular sector $\sigma_{bct}$), when defining $F_p$ for a given point $p$ in the $(m, d)$-space.
		(A) $L_1 : y = m_1x + d_1$,
		(B) $L_2 : y = m_2x + d_2$,
		(C) $L_3 : y = m_3x + d_3$, and
		(D) $L_4 : y = m_4x + d_4$,
		where $\{ (m_i, d_i) | 1 \leq i \leq 4\}$ is the set of corner points of $F_p$ as illustrated in Figure \ref{fig_subprob4_3}.}
	\label{fig_subprob4_4}
\end{figure}

Note that $\eta_p(m, d)$ is a partially defined function, where $m = -1/\tan \theta$, $d = r/\sin \theta$, $\theta \in [0, 2\pi)$, and $r \in [0, R]$.
We define a collection $C_P = \{ \eta_p | p \in P \}$ of $n$ bivariate functions, and construct the maximization diagram $M_P$ (i.e., upper envelope) of $C_P$.
Upper envelope $M_P$ can be computed in $O(n^{2+\epsilon})$ time using $O(n^{2+\epsilon})$ space, for any $\epsilon > 0$ \cite{agarwal96overlay}.

Given a query line $L$, recall that the objective is to find the point $b^*$ closest to $b$ on $L^-$ such that the ``sector'' bounded by $bc^*$, $bt$, and $\gamma_{c^*t}$ is free of $P$.
Let $m$ and $d$ be the slope and $y$-intercept of $L$, respectively.
Since $M_P = \| b - b^* \|$, $b^*$ can be easily determined after looking up $(m, d)$ in $M_P$ in $O(\log n)$ time (with the help of a supporting point location query data structure).

Note that a similar analysis also applies to the case of line segments.
The only difference is that the characterization of $\eta_p$ (as well as $F_p$) involves the tangency point between a circle (centered at $b'$ and passing through $t$) and a given line segment in addition to the endpoints of the line segment.
Nonetheless, the resulting expressions (for $\eta_p$ and $F_p$) are still algebraic and of constant degree, and hence $M_P$ can be constructed using the same time and space complexities.

\begin{lem}
	In Subproblem \ref{subprob4}, a set $P$ of $n$ line segments and a fixed origin $t$ can be preprocessed in $O(n^{2+\epsilon})$ time using $O(n^{2+\epsilon})$ space, for any $\epsilon > 0$, so that one can find, in $O(\log n)$ time, the point $b^* \in L^-$ closest to $b$ such that the ``sector'' bounded by $bc^*$, $bt$, and $\gamma_{c^*t}$ is free of $P$.
\end{lem}

\section{Computing all feasible values of \boldmath$r$}
\label{sec:all_r}

%We now extend the algorithms described in the previous section to solve Problem \ref{prob_all_r}.
In this section, we address Problem \ref{prob_all_r}.
We begin by providing a version of Lemma \ref{lem1} for a feasible articulated trajectory with the \emph{maximum length $r$}.

\begin{lem}
	\label{lem2}
	For a feasible max-$r$ articulated trajectory, there is a corresponding extremal feasible max-$r$ articulated trajectory such that, in its final configuration, $bt$ passes through an obstacle endpoint, and at least one of the following is true:
	\renewcommand\labelenumi{\Roman{enumi})}
	\renewcommand\theenumi\labelenumi
	\begin{enumerate}
		\item $ab$ passes through two obstacle endpoints,
		\item $ab$ passes through an obstacle endpoint, and $\angle cbt = \pi/2$ radians (or $r = R$, whichever occurs first),
		\item $ab$ passes through an obstacle endpoint, and $\gamma_{ct}$ intersects an obstacle endpoint or is tangent to an obstacle line segment,
		\item $ab$ and $bc$ each pass through an obstacle endpoint, or
		\item $ab$ passes through an obstacle endpoint, and $bc$ intersects an obstacle line segment at $c$.
	\end{enumerate}	
\end{lem}

\begin{proof}
	The lemma follows immediately from the proof of Lemma \ref{lem1}.
\end{proof}

We know from \cite[Lemma 2.1]{teo20traj} that, for any given $r$, a feasible articulated trajectory is always associated with an ``extremal'' feasible articulated trajectory with a final configuration in which i) $ab$ passes through two obstacle endpoints, or ii) $ab$ and $bt$ each pass through an obstacle endpoint.
Suppose that, for each pair of obstacle endpoints $u, v \in V$, we consider the extremal feasible articulated trajectory intersecting $u$ and $v$ in the way just described (if one exists), and we characterize the range of $r$ for which the extremal trajectory remains feasible (i.e., by varying the length $r$ of segment $bc$ of the probe while maintaining the intersections of the trajectory with $u$ and $v$).
In fact, this is exactly what was illustrated in the proof of Lemma \ref{lem1} (and stated in Corollary \ref{cor}).
The resulting ranges of $r$ form the feasible domain of $r$.
Observe that the lower and upper limits of these feasible ranges of $r$ are given by the extremal feasible min- and max-$r$ articulated trajectories.
Consequently, based on Lemmas \ref{lem1} and \ref{lem2}, by using the enumerate-and-verify approach described in Section \ref{enum_and_verify}, if we were to choose carefully an ordering (of the obstacles) in the enumeration of the $O(n^3)$ extremal feasible min- and max-$r$ articulated trajectories, we could produce the set of all feasible values of $r$.
Hence, we obtain the following result.

\begin{thm}
	The feasible domain of $r$ can be computed in $O(n^3 \log n)$ time using $O(n^{2+\epsilon})$ space, for any constant $\epsilon > 0$.
\end{thm}

We now consider an approach to solving Problem \ref{prob_all_r} by following the perturbation steps in the proof of Lemma \ref{lem1}.
Observe that the sequence of trajectory perturbations described in Scenario A in the proof of Lemma \ref{lem1} is essentially analogous to the procedure described in Section \ref{sec:min_r} for finding an extremal feasible min-$r$ articulated trajectory.
As a result, the range of $r$ associated with the segment $b''b'''$ found for a pair of obstacle endpoints $u, v \in V$, where $u \neq v$, in that procedure corresponds to a maximal contiguous subset of the feasible values of $r$ for the pair $(u, v)$.
Thus, in order to find the entire feasible domain of $r$, we have to additionally address Scenario B in the proof of Lemma \ref{lem1}, as detailed next.

\begin{figure}[h]
	\centering
	\includegraphics[scale=0.16]{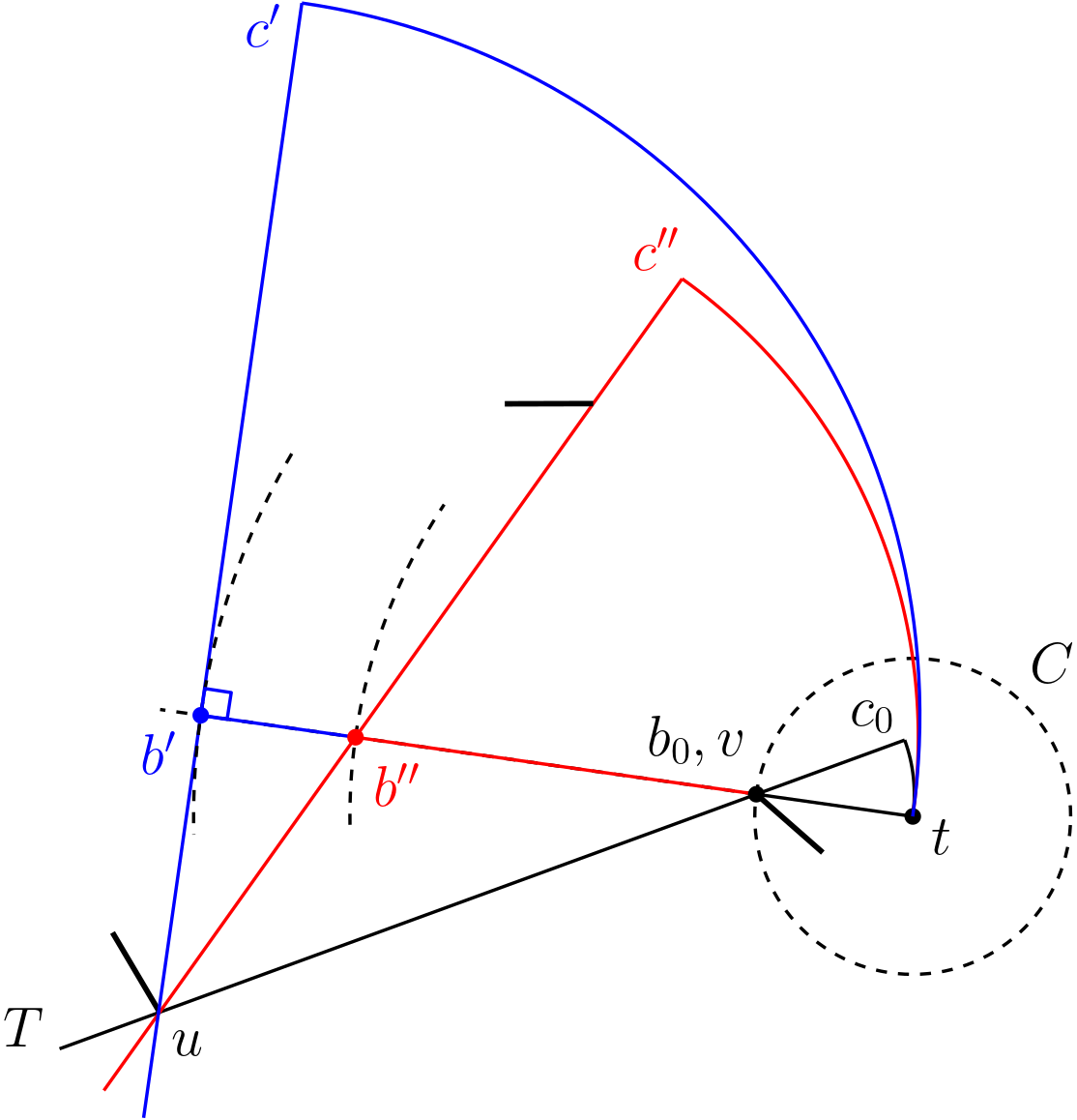}
	\caption{Illustration for the sequence of steps to find $[r_\alpha, r_\beta]$, where $r_\alpha = |b_0c_0|$ and $r_\beta = |b''c''|$.}
	\label{fig_steps_B}
\end{figure}

For a pair of obstacle endpoints $u, v \in V$, where $u \neq v$, let $T$ denote an articulated probe trajectory whose $ab$ passes through $u$, and $bt$ passes through $v$.
Without loss of generality, assume that $T$ rotates line segment $bc$ (of the probe) clockwise around $b$ to reach $t$ (Figure \ref{fig_steps_B}).
Note that, for such a pair $(u,v)$ and a trajectory $T$, we may have multiple disjoint contiguous feasible ranges of $r$ (see the proof of Lemma 2).
This situation will be addressed later in the section.
Now, in the following sequence of steps, we look for the contiguous range containing the smallest feasible value of $r$ for the pair $(u, v)$ -- that is, the contiguous feasible range of $r$ whose lowest value is the closest overall to the length of $|bt|$ when $b = v$.
We denote the resulting range as $[r_\alpha, r_\beta]$.

\begin{enumerate}
	\renewcommand\labelenumi{B\arabic{enumi}.}
	\renewcommand\theenumi\labelenumi
	
	\item Let $a_0$, $b_0$, and $c_0$ denote the positions of $a$, $b$, and $c$ (of the probe), respectively, when $b = v$.
	Check if $a_0b_0$ intersects any obstacle.
	If it does, we stop (this scenario will be taken care of later).
	Otherwise, proceed with step B2.
	
	\item Let $\gamma_{c_0t}$ be the circular arc centered at $b_0$ and emanating counter-clockwise from $t$ to $c_0$.
	Let $\sigma_{b_0c_0t}$ be the circular sector bounded by $b_0c_0$, $b_0t$, and $\gamma_{c_0t}$.
	Check if $\sigma_{b_0c_0t}$ intersects any obstacle. 
	If it does, then there is no feasible articulated trajectory whose $ab$ passes through $u$, and $bt$ passes through $v$.
	Otherwise, set $r_\alpha = |b_0c_0|$, and proceed with step B3.
	
	\item Let $b'$ denote the position of $b$ when $\angle cbt = \pi/2$ radians.
	Find the closest point $b'' \in b'b_0$ to $b_0$ such that one of the following happens:
	i) either $a''u$ or $ub''$ passes through an obstacle endpoint, 
	ii) circular arc $\gamma_{c''t}$, centered at $b''$ and emanating counter-clockwise from $t$ to $c''$, intersects an obstacle endpoint or is tangent to an obstacle line segment, 
	iii) $b''c''$ passes through an obstacle endpoint, or 
	iv) $b''c''$ intersects an obstacle line segment at $c''$, 
	where $a''$ and $c''$ are the positions of $a$ and $c$, respectively, when $b = b''$.
	Set $r_\beta = |b''c''|$.
\end{enumerate}

Each of these steps can be resolved by simply checking against each of the $O(n)$ obstacles, resulting in an $O(n^3)$-time algorithm for finding $[r_\alpha, r_\beta]$ for all pairs of obstacle vertices $u,v \in V$.
This can be improved upon by using efficient query data structures (see Table \ref{tab2} for a summary), whose discussion is deferred to Section \ref{query_B} below.
Recall that $O(n^2)$ queries are to be answered in the worst case.
Hence, a total of $O(n^{5/2})$ time is required to find $[r_\alpha, r_\beta]$ for all pairs of obstacle endpoints $u, v \in V$.

\renewcommand{\arraystretch}{1.5}
\begin{table}
	\footnotesize
	\centering
	\caption{Summary of query data structures used in steps B1-B3.
		The size, preprocessing time, and query time of a data structure are denoted by $S(n)$, $P(n)$, and $Q(n)$, respectively.
		Our results are highlighted in gray.}
	\vspace*{3mm}
	\label{tab2}
	\begin{tabular}{ l >{\raggedright\arraybackslash} p{4.8cm} l l l }
		\hline
		\textbf{Step} & \textbf{Query} & \boldmath$S(n)$ & \boldmath$P(n)$ & \boldmath$Q(n)$ \\
		\hline
		
		B1 & Ray shooting queries \cite{pocchiola90graphics} &
		$O(n^2)$ & $O(n^2)$ & $O(\log n)$ \\
		
		\rowcolor{Gray}
		B2 & Subproblem \ref{subprob1}: Circular sector emptiness queries &
		$O(n^{2+\epsilon})$ & $O(n^{2+\epsilon})$ & $O(\log n)$ \\
		
		\rowcolor{Gray}
		B3 (i) & Subproblem \ref{subprob5} &
		$O(n^2)$ & $O(n^{2+\epsilon})$ & $O(n^{1/2})$ \\
		
		\rowcolor{Gray}
		B3 (ii), (iii), (iv) & Subproblem \ref{subprob6} &
		$O(n^2)$ & $O(n^2 \log n)$ & $O(\log n)$ \\ 
		
		\hline
	\end{tabular}
\end{table}

\paragraph{Multiple disjoint feasible ranges of \boldmath$r$ in Scenario B}
We now consider the following situation.
Suppose that $r_\beta$ corresponds to the occurrence of either case (ii), (iii), or (iv) in step B3.
Then, $r_\beta < r \leq r_{\pi/2}$ is certainly infeasible (see the proof of Lemma \ref{lem1}), where $r_{\pi/2}$ denotes the length of line segment $bc$ when $\angle bct = \pi/2$ radians.
On the other hand, suppose that $r_\beta$ is obtained as a result of case (i) in step B3 -- that is, $ab$ passes through an endpoint $p$ of an obstacle line segment $s$ (Figure \ref{fig_sit}).
Let $q$ denote the other endpoint of $s$.
In this case, if we continue to increase $r$ beyond $r_\beta$ (by moving $b$ further away from $b''$ and toward $b'$), the trajectory may become feasible again as line segment $ab$ passes through $q$ of $s$.
Thus, this phenomenon may lead to multiple disjoint maximal contiguous feasible subsets of $[r_\alpha, r_{\pi/2}]$.

\begin{figure}[h]
	\centering
	\includegraphics[scale=0.16]{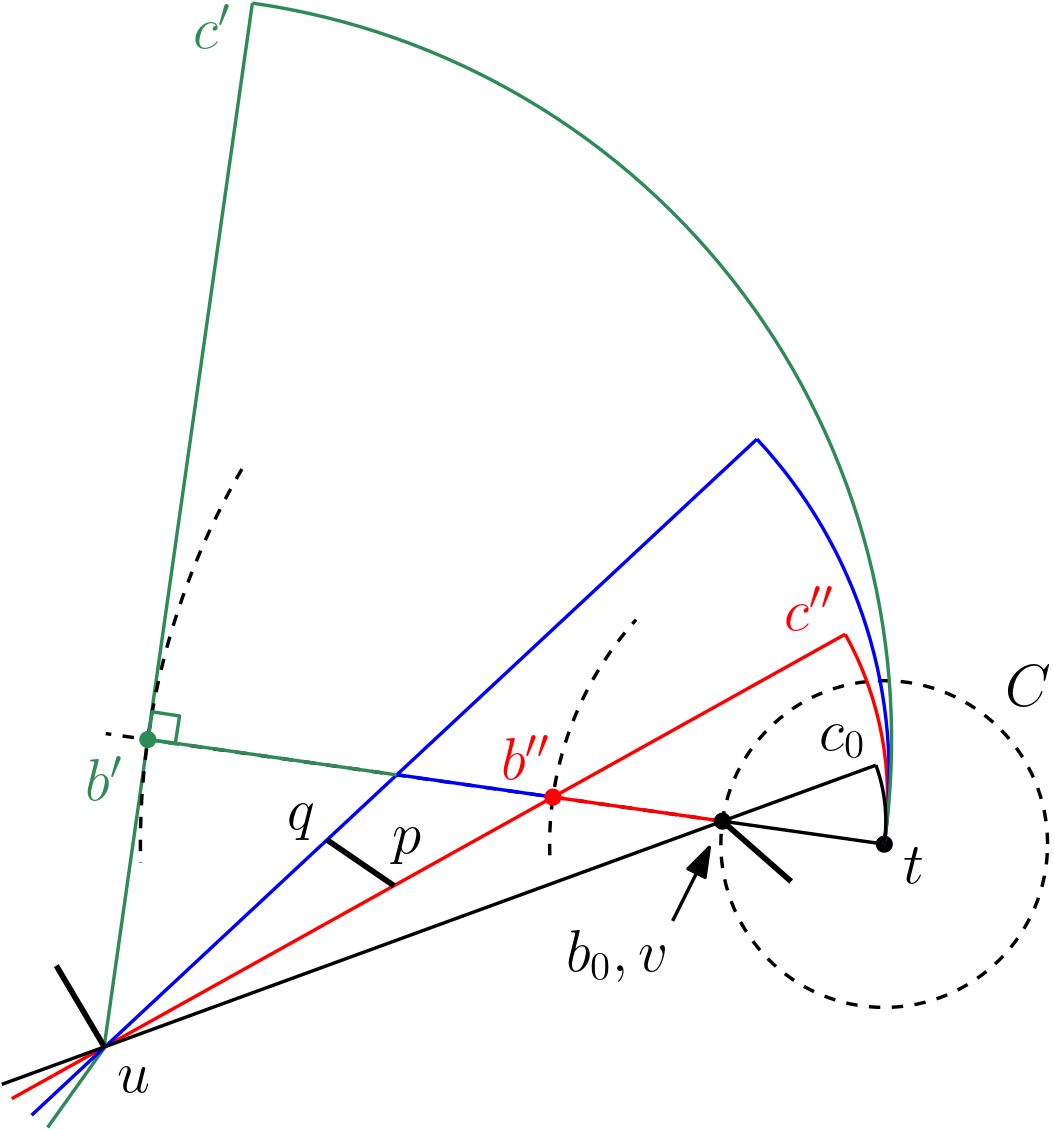}
	\caption{Disjoint contiguous feasible subsets of $[r_\alpha, r_{\pi/2}]$, where $r_\alpha = |b_0c_0|$ and $r_{\pi/2} = |b'c'|$.}
	\label{fig_sit}
\end{figure}

This situation, as it turns out, can be easily handled as follows.
First, notice that we have encountered, in the procedure described in Section \ref{sec:min_r} (specifically, step A6), the extremal feasible articulated trajectory whose $ab$ passes through $u$ and $q$, and $bc$ passes through $v$.
Hence, in order to continue characterizing the feasible range of $r$ in this case, we perform A6s as an additional step after A6.
Note that step A6s is the same as B3, but with $b_0$ replaced by the point of intersection between $b'b''$ and the supporting line of $uq$.
The range of $r$ associated with line segment $b''b^\#$ found in step A6s is a maximal contiguous feasible $r$-interval (for the pair $(u, w)$ defined in step A6s).
%B3 ($b_0$ replaced by the point of intersection between $b'b''$ and the supporting line of $uq$) upon encountering said extremal trajectory in step A6 (in other words, while performing steps A1 through A7, we include B3 as an additional step after A6).
In Step B1, the case of $a_0b_0$ intersecting an obstacle line segment is similar to what has just been described, and thus can be resolved likewise. \\

We now give an upper bound on the total number of feasible contiguous ranges of $r$ (either disjoint or overlapping) computed in the procedure just described for finding the feasible domain of $r$.
In steps A1-A7 (including A6s), we find at most two feasible contiguous ranges of $r$ for each pair of points in $P$.
In steps B1-B3, we find at most one feasible contiguous range of $r$ for every pair of points in $P$.
So, the total number of feasible contiguous ranges of $r$, computed over all pairs of points in $P$ using the procedure, must be bounded by $O(n^2)$.

All in all, after considering the time and space complexities of each routine involved, we reach the following conclusion.

\begin{thm}
	All values of $r$ for which at least one feasible trajectory exists can be determined in $O(n^{5/2})$ time using $O(n^{2+\epsilon})$ space, for any constant $\epsilon > 0$.
\end{thm}

Note that the procedure described in this section also gives an implicit characterization of the ``extremal'' feasible articulated trajectory space for all feasible values of $r$.

\subsection{Query problems in steps B1-B3}
\label{query_B}

\paragraph{Step B1 (ray shooting queries)}
Refer to \cite{pocchiola90graphics}.

\paragraph{Step B2 (circular sector emptiness queries)}
See Subproblem \ref{subprob1} in Section \ref{query_A}.

\paragraph{Step B3}
For case (i) of step B3, we have the following query problem.

\begin{subprob}
	\label{subprob5}
	Given a fixed origin $t$, a fixed point $u$, and a set $P$ of $n$ point, let $L$ be the line passing through $t$.
	Let $ub$ be the perpendicular line segment dropped from $u$ to $L$, where $b \in L$.
	Line $L$ divides the plane into two half-planes $H^+$ and $H^-$.
	Without loss of generality, assume that $u$ is located in $H^+$.
	For a query point $v \in bt$, let $m_v$ be the slope of the supporting line of $uv$.
	Let $q$ denote a point in $P \cap H^+$, and $m_q$ be the slope of the supporting line of $uq$.
	Preprocess $t$, $u$, and $P$ so that, given a point $v \in bt$ at query time, one can efficiently find the point $q \in P \cap H^+$ that has the smallest slope $m_q$ greater than $m_v$.
\end{subprob}

As with Subproblem \ref{subprob3}, we can construct a data structure based on the half-space decomposition scheme as described by Matou{\v{s}}ek in \cite[Theorem 5.1]{matouvsek93range}.
The only difference is that, for each inner subset, we store the points $p$ in the sorted order according to their slopes $m_p$, where $m_p$ is the slope of the supporting line of $up$, so that a binary search query can be performed in $O(\log n)$ time to find the point in the inner subset that has the smallest slope greater than $m_v$.
In the end, we obtain a $O(\log^2 n)$ query-time data structure of size $O(n^2 / \log^2 n)$, and it can be built in $O(n^2 / \log^2 n)$ time.

We now consider creating a trade-off between space and time usage by our data structure by following the strategy outlined in \cite[Theorem 6.2]{matouvsek93range}.
For any $n \leq m \leq n^2$, there exists a half-space decomposition scheme with $O(m)$ space, $O(n^{1+\epsilon} + m \log^{\epsilon} n)$ preprocessing time, and $O(n/m^{1/2})$ query time.
Thus, we can set $m = n$, and obtain a query data structure with $O(n)$ space, $O(n^{1+\epsilon})$ preprocessing time, and $O(n^{1/2})$ query time.

\begin{lem}
	In Subproblem \ref{subprob5}, for a fixed origin $t$ and a fixed point $u$, a set $P$ of $n$ points can be preprocessed in $O(n^{1+\epsilon})$ time into a data structure of size $O(n)$ so that, given  a query point $v \in bt$, one can determine the point $q \in P \cap H^+$ that has the smallest slope $m_q$ greater than $m_v$ in $O(n^{1/2})$ time.
\end{lem}

In step B3, cases (ii), (iii), and (iv) can be addressed as follows.

\begin{subprob}
	\label{subprob6}
	Given a fixed origin $t$ and a fixed point $u$, for a query line $L$ that passes through $t$, let $ub$ be the perpendicular line segment dropped from $u$ to line $L$, where $b \in L$ (Figure \ref{fig_subprob6}).
	Note that $L$ divides the plane into two half-planes $H^+$ and $H^-$.
	Without loss of generality, assume that $u$ is located in $H^+$.
	Let $b^*$ be a point on $bt$.
	Let $\gamma_{c*t}$ denote the acute circular arc centered at $b^*$ and emanating from $t$ to $c^*$, where $c^*$ is the point of intersection, in $H^-$, between the supporting line of $ub^*$ and the circle of radius $|b^*t|$ centered at $b^*$.
	Let $\sigma_{b^*c^*t}$ be the circular sector bounded by $b^*c^*$, $b^*t$, and $\gamma_{c*_t}$.
	For a fixed origin $t$, a fixed point $u$, and a given set $P$ of $n$ line segment, preprocess them so that, given a line $L$ passing through $t$ at query time, one can efficiently find the point $b^*$ farthest from $t$ on $bt$ such that $\sigma_{b^*c^*t}$ is free of $P$.
\end{subprob}

\begin{figure}[h]
	\centering
	\includegraphics[scale=0.16]{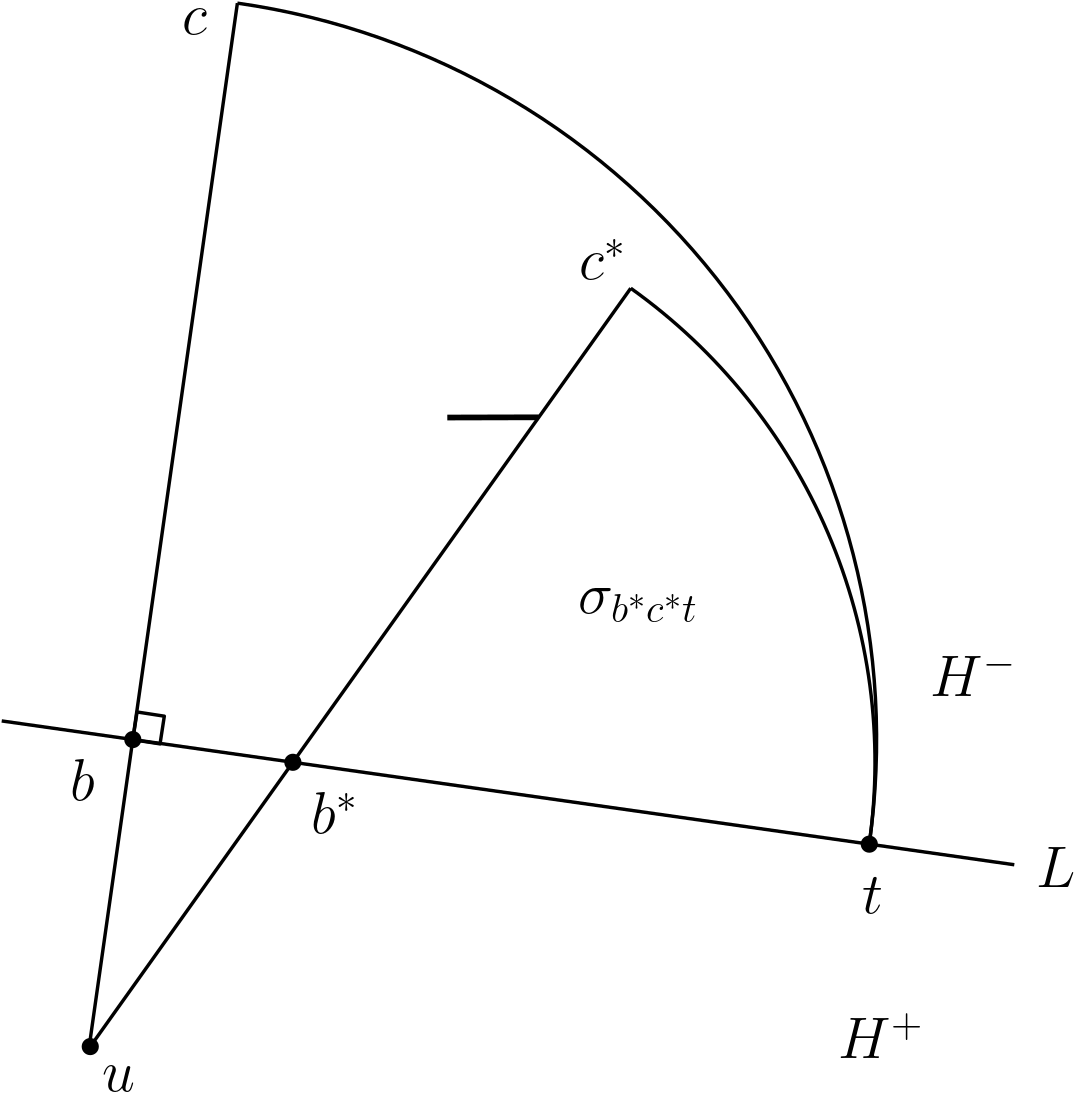}
	\caption{Illustration of the farthest point $b^*$ from $t$ on $bt$ such that circular sector $\sigma_{b^*c^*t}$ is free of line segments.}
	\label{fig_subprob6}
\end{figure}

The returned solution of a query in Subproblem \ref{subprob6} (i.e., point $b^*$) may be associated with one of the following three intersection cases --
1) $\gamma_{c^*t}$ intersects an obstacle endpoint or is tangent to obstacle line segment,
2) $b^*c^*$ passes through an obstacle endpoint, or
3) $b^*c^*$ intersects an obstacle line segment at $c^*$.

In each of these cases, we can define a collection $C$ of $n$ univariate partial functions, and construct the minimization diagram $M$ of $C$.
Using lower envelope $M$, we can find the closest point $b' \in bt$ to $t$ for which the intersection associated with the case occurs (e.g., in case 1, the point $b' \in bt$ closest to $t$ such that the corresponding circular arc $\gamma_{c't}$ intersects an obstacle endpoint or is tangent to obstacle line segment).
Finally, $b^*$ is the closest point $b' \in bt$ to $t$ of all three cases.

We can construct the lower envelope in each of the three cases as follows.
The expression for a query line $L$ passing through the origin $t$ is given by $y = mx$, where $m$ is the slope of $L$.
Let $\gamma_{ct}$ denote the quarter circular arc centered at $b$ and emanating from $t$ to $c$, where $c$ is the point of intersection, in $H^-$, between the supporting line of $ub$ and the circle of radius $|bt|$ centered at $b$.
Let $\sigma_{bct}$ be the quarter circular sector bounded by $bc$, $bt$, and $\gamma_{ct}$.
Note that $\sigma_{bct}$ is uniquely determined by $m$.

\subparagraph{Case 1.}
For a line segment $s \in P$ located within $\sigma_{bct}$, there exists a circle $D_p$ that i) passes through $t$, ii) is tangent to $s$ at a point $p$, and iii) is centered at a point $b'$ on $bt$.
Note that at most three circles $D_p$ are defined for $s$, since such a circle $D_p$ may pass through one of the two endpoints of $s$, or may be tangent to $s$ at an interior point of $s$.
We define $f_p(m)$ to be the distance from $b'$ to $t$.
In addition, $f_p(m)$ is only defined when $p$ is located on $\gamma_{c't}$, which is the circular arc of $D_p$ emanating from $t$ to $c'$, where $c'$ is the point of intersection, in $H^-$, between $D_p$ and the supporting line of $ub'$.
Note that $f_p$ is partially defined as a function of $m$.
Moreover, given any two distinct points $p$ and $q$, $f_p(m)$ = $f_q(m)$ occurs for at most one value of $m$ (due to the fact that three points – $p$, $q$, and $t$ – define a unique circle).
Let $M$ be the lower envelope of $f_p(m)$ for all line segments $s \in P$.

\subparagraph{Case 2.}
For a line segment $s \in P$ located within $\sigma_{bct}$, let $b'$ be the point of intersection between $bt$ and the supporting line of $up$, where $p$ is an endpoint of $s$.
We define $f_p(m)$ as the distance from $b'$ to $t$.
Furthermore, $f_p(m)$ is only defined when $|b'p| \leq |b't|$.
Thus, $f_p$ is a partially defined function of $m$.
For any two distinct endpoints $p$ and $q$, the corresponding curves $f_p(m)$ and $f_q(m)$ intersects at most once (since a line is defined by two points).
Let $M$ denote the lower envelope of $f_p(m)$ for all $s \in P$.

\subparagraph{Case 3.}
For a line segment $s \in P$ located within $\sigma_{bct}$, let $b' \in bt$ be the center of a circle $D$ passing through $t$, such that circle $D$ and the supporting line of $ub'$ intersect at a point on $s$.
As before, we define $f_s(m)$ as the distance from $b'$ to $t$.
Function $f_s$ is partially defined over $m$.
For any two disjoint line segments $s_i$ and $s_j$, the corresponding curves $f_{s_i}(m)$ and $f_{s_j}(m)$ only intersect in at most one point (since $s_i$ and $s_j$ may only intersect at their endpoints).
Let $M$ be the lower envelope of $f_s(m)$ for all $s \in P$. \\

In each case above, given the properties that each pair of the partially defined functions only intersect at most once, the lower envelope $M$ can be constructed in $O(n \log n)$ time using $O(n)$ space \cite{sharir95dav}.
Given a query line $L$ (of a slope $m$) passing through $t$, we can perform, in each case, a binary search on the respective lower envelope $M$ in $O(\log n)$ time, and find the point $b' \in bt$ closest to $t$ such that the intersection associated with the case occurs.
%$b^*$ is the closest point $b' \in bt$ to $t$ of all four cases.
We can now conclude the following.

\begin{lem}
	In Subproblem \ref{subprob6}, for a fixed origin $t$ and a fixed point $u$, a set of $P$ of $n$ line segments can be preprocessed in $O(n \log n)$ time into a data structure of size $O(n)$ so that, for a query line $L$ that passes through $t$, one can find, in $O(\log n)$ time, the point $b^* \in bt$ farthest from $t$ such that $\sigma_{b^*c^*t}$ is free of $P$.
\end{lem}

\paragraph{Final note on step B3} 
Recall that $V$ denotes the set of obstacle endpoints of $P$.
For each point $p \in V$, we construct the two data structures $\Gamma_p$ and $\Pi_p$ required for solving Subproblems \ref{subprob5} and \ref{subprob6}, respectively, by using $p$ as the fixed point $u$ in the subproblems.
Thus, the total preprocessing times required to compute $\Gamma_p$ and $\Pi_p$ for all points $p \in V$ are bounded by $O(n^{2+\epsilon})$ and $O(n^2 \log n)$, respectively, and the space usage is $O(n^2)$.
In step B3, for a given pair of points $u, v \in V$, we use the data structure $\Gamma_u$ built for point $u$ to check for case (i) in $O(n^{1/2})$ time, and $\Pi_u$ for cases (ii), (iii), and (iv) in $O(\log n)$ time.

Note that the point $q$ returned by the query data structure $\Gamma_u$ is only considered valid if the supporting line of $qu$ intersects $bt$.
Let $b_\text{i}$ denote the point of intersection between $qu$ and $bt$.
In addition, let $b_\text{ii,iii,iv}$ denote the point $b^*$ returned by $\Pi_u$.
Point $b''$, as defined in step B3, is the closest of $b_\text{i}$ and $b_\text{ii,iii,iv}$ to $b_0$.
Since the query time of case (i) is dominant over those of cases (ii), (iii), and (iv), we conclude that step B3 takes $O(n^{1/2})$ time.

\section{Characterizing feasible trajectory space}
\label{sec:feas_space}

In this section, we describe our solution to Problem \ref{prob_traj_space}.
We begin by explicitly characterizing the following for a given length $r$:
i) the final configuration space,
ii) the \emph{forbidden} final configuration space, and
iii) the \emph{infeasible} final configuration space.

\subsection{Final configuration space}
In a final configuration of the articulated probe, point $a$ can be assumed to be on $S$, and point $b$ lies on the circle $C$ of radius $r$ centered at $t$ (Figure \ref{fig_traj}).
Let $\theta_S$ and $\theta_C$ be the angles of line segments $ta$ and $tb$ measured counter-clockwise from the $x$-axis, where $\theta_S, \theta_C \in [0, 2\pi)$.
Since $bc$ may rotate around $b$ as far as $\pi/2$ radians in either direction, for any given $\theta_S$, we have $\theta_C \in [\theta_S - \cos^{-1} r/R, \theta_S + \cos^{-1} r/R]$.
We call this the \emph{unforbidden} range of $\theta_C$.
A final configuration of the articulated probe can be specified by ($\theta_S$, $\theta_C$), depending on the locations of points $a$ and $b$ on circles $S$ and $C$, respectively (Figure \ref{fig_fin}).
The final configuration space $\Sigma_{fin}$ of the %articulated
probe can be computed in $O(1)$ time.

\begin{figure}
	\centering
	\includegraphics[scale=0.135]{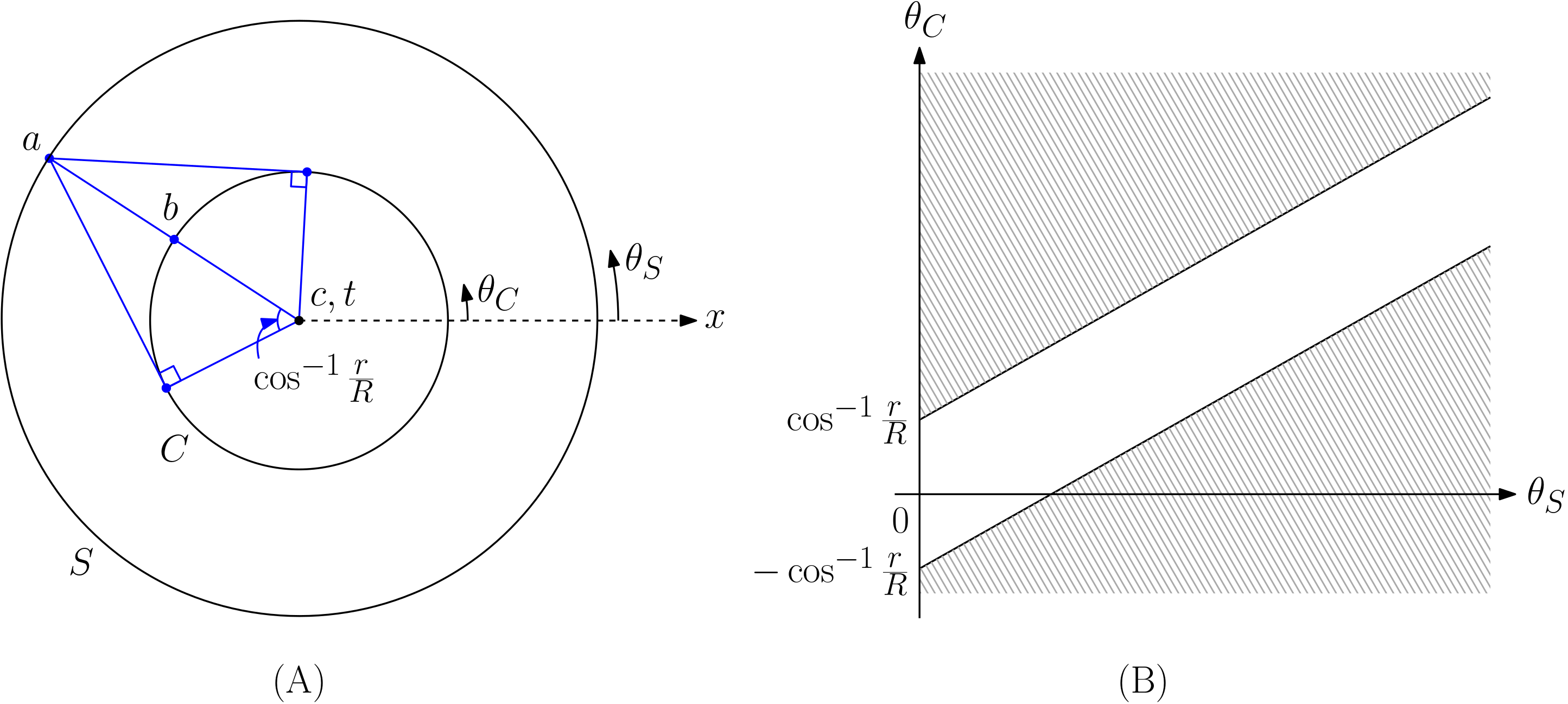}
	\caption{Final configurations of the articulated probe.
		(A) Each value of $\theta_S$ is associated with an unforbidden range of $\theta_C$ spanning from $\theta_S - \cos^{-1} r/R$ to $\theta_S + \cos^{-1} r/R$.
		(B) The unshaded region of the $(\theta_S, \theta_C)$-plot represents the unforbidden final configuration space when $S$ is obstacle-free.}
	\label{fig_fin}
\end{figure}

\subsection{Forbidden final configuration space}
\label{subsec:forb_space}
A final configuration is called \emph{forbidden} if the final configuration (represented by $ab$ and $bt$) intersects one or more of the obstacle line segments.
Let $s$ be an obstacle line segment of $P$.
We have two different cases, depending on whether $s$ is located 1) outside or 2) inside $C$.

\paragraph{Case 1. Obstacle line segment \boldmath$s$ outside \boldmath$C$}
%The corresponding forbidden final configuration space can be characterized as follows.
Let angles $\theta_i$, where $i = 1, \dots, 6$, be defined in the manner depicted in Figure \ref{fig_forb_out}A.
Briefly, each $\theta_i$ corresponds to an angle $\theta_S$ at which point a tangent line i) between $C$ and $s$ or ii) from $t$ to $s$, intersects $S$.
As $\theta_S$ increases from $\theta_1$ to $\theta_3$, the upper bound of the unforbidden range of $\theta_C$ decreases as a continuous function of $\theta_S$.
Similarly, when $\theta_S$ varies from $\theta_4$ to $\theta_6$, the lower bound of the unforbidden range of $\theta_C$ decreases as a continuous function of $\theta_S$.
For $\theta_3 \leq \theta_S \leq \theta_4$, there exists no unforbidden final configuration at any $\theta_C$ (Figure \ref{fig_forb_out}B).
For conciseness, the upper (resp. lower) bound of the unforbidden range of $\theta_C$ is referred to as the upper (resp. lower) bound of $\theta_C$ hereafter.

\begin{figure}
	\centering
	\includegraphics[scale=0.135]{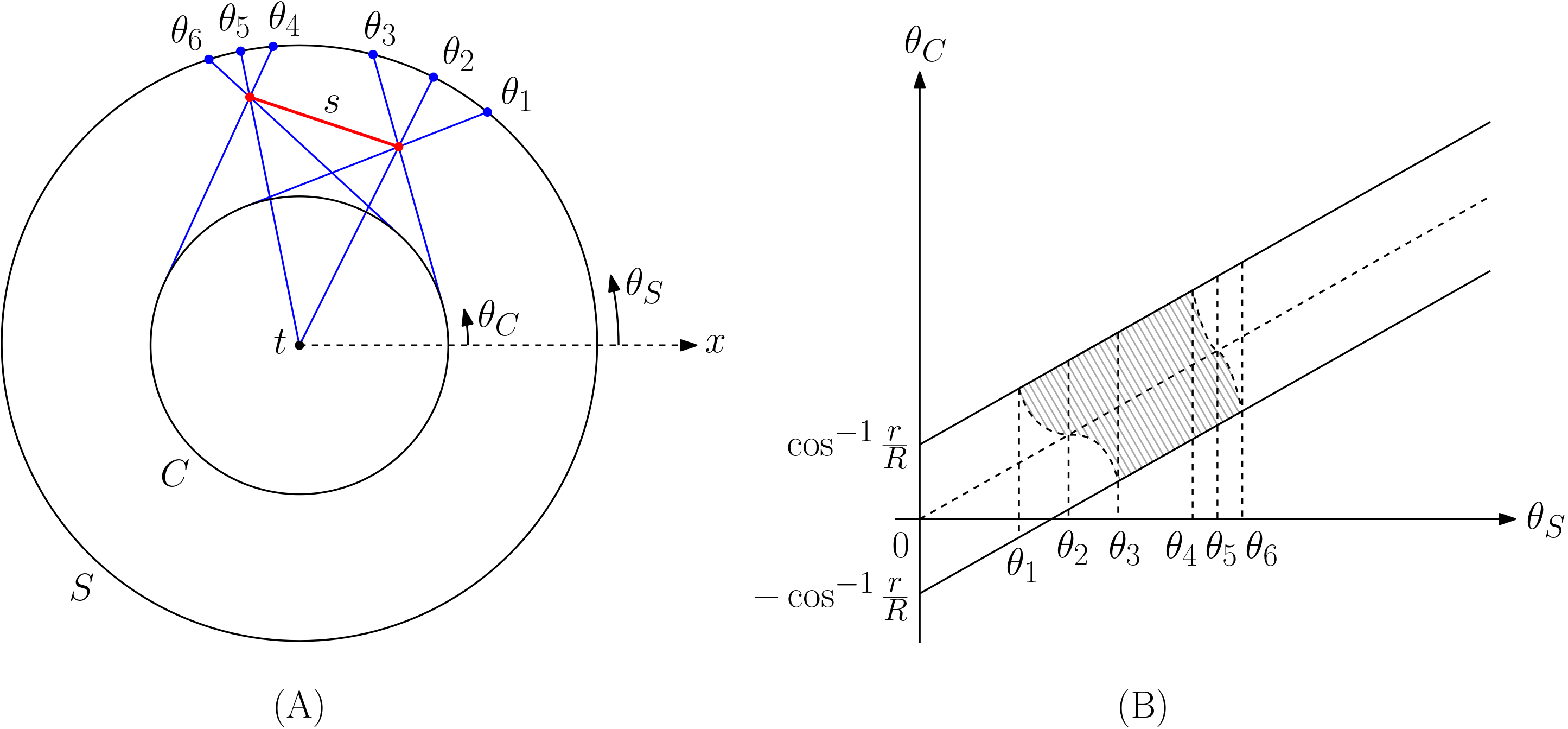}
	\caption{Forbidden final configurations due to an obstacle line segment $s$ outside $C$.}
	\label{fig_forb_out}
\end{figure}

\paragraph{Case 2. Obstacle line segment \boldmath$s$ inside \boldmath$C$}
We can similarly compute the forbidden final configuration space for an obstacle line segment $s$ inside $C$.
%Note that, as shown in
Note in Figure \ref{fig_forb_in}A that angles $\theta_i$, where $i = 1, \dots, 6$, are defined differently from case 1.
For $\theta_1 \leq \theta_S \leq \theta_4$, the upper bound of $\theta_C$ is equivalent to $\theta_2$.
For $\theta_3 \leq \theta_S \leq \theta_6$, the lower bound of $\theta_C$ equals to $\theta_5$ (Figure \ref{fig_forb_in}B). \\

\begin{figure}
	\centering
	\includegraphics[scale=0.135]{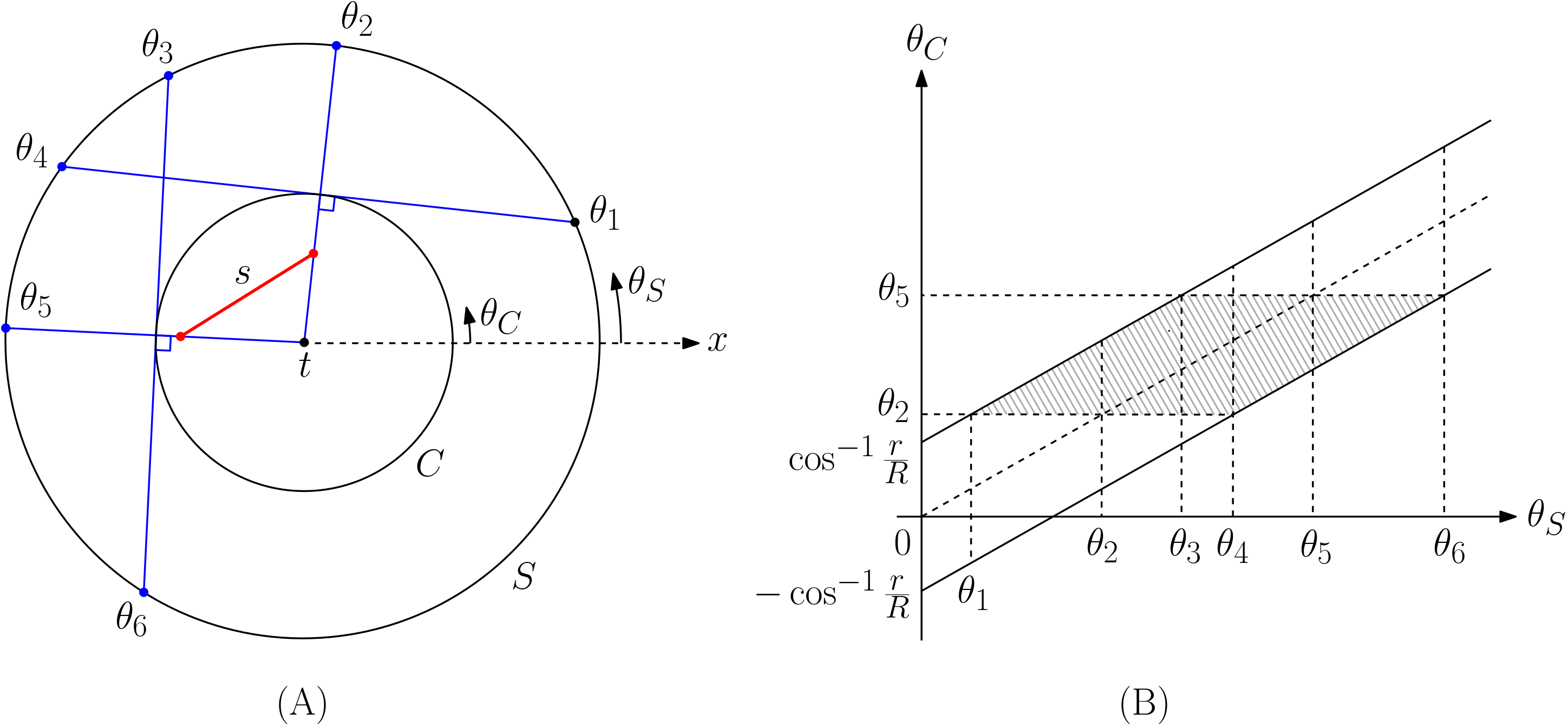}
	\caption{Forbidden final configurations due to an obstacle line segment $s$ inside $C$.}
	\label{fig_forb_in}
\end{figure}

We can find the forbidden final configuration space for an obstacle line segment
%(i.e., final configuration obstacle)
in $O(1)$ time.
Thus, for $n$ obstacle line segments, it takes $O(n)$ time to compute the corresponding set of forbidden final configurations.
The union of these configurations forms the forbidden final configuration space $\Sigma_{fin, forb}$ of the articulated probe.
The free final configuration space
%$\Sigma_{fin, free}$ 
of the articulated probe is
%the complement of $\Sigma_{fin, forb}$; that is, 
$\Sigma_{fin, free} = \Sigma_{fin} \setminus \Sigma_{fin, forb}$.

\subsection{Infeasible final configuration space}
\label{inf_fin_config}
The feasible trajectory space of the articulated probe can be characterized as a subset of $\Sigma_{fin, free}$.
A final configuration is called \emph{infeasible} if the circular sector associated with the final configuration (i.e., the area swept by line segment $bc$ of the probe to reach $t$) intersects any obstacle line segment.
We denote the infeasible final configuration space as $\Sigma_{fin, inf}$.

Let $C'$ be the circle centered at $t$ and of radius $\sqrt{2}r$.
A circular sector associated with a final configuration can only intersect an obstacle line segment lying inside $C'$.
%Recall that segment $bc$ may rotate as far as $\pi/2$ radians in either direction.
%Thus, for any given $\theta_C$, we have a \emph{feasible} range of $\theta_S \in [\theta_C - \cos^{-1} r/R, \theta_C + \cos^{-1} r/R]$ (provided that the workspace is free of obstacles).
Instead of characterizing the lower and upper bounds of $\theta_C$ as $\theta_S$ varies from 0 to $2\pi$ (as in Section \ref{subsec:forb_space}), here we perform the characterization the other way around.
For conciseness, we only present arguments for the negative half of the $\theta_S$-range, which is $[\theta_C - \cos^{-1} r/R, \theta_C]$; similar arguments apply to the other half due to symmetry.
We have two cases, depending on whether an obstacle line segment $s$ lies 1) inside $C$ or 2) outside $C$ and inside $C'$.

\paragraph{Case 1. Obstacle line segment \boldmath$s$ inside \boldmath$C$}
For brevity, the quarter circular sector associated with a point $b$ (i.e., the maximum possible area swept by line segment $bc$ of the probe to reach $t$), where the angle of $tb$ (relative to the $x$-axis) is $\theta_C$, is referred to as the \emph{quart-sector of $\theta_C$}.

\begin{figure}
	\centering
	\includegraphics[scale=0.135]{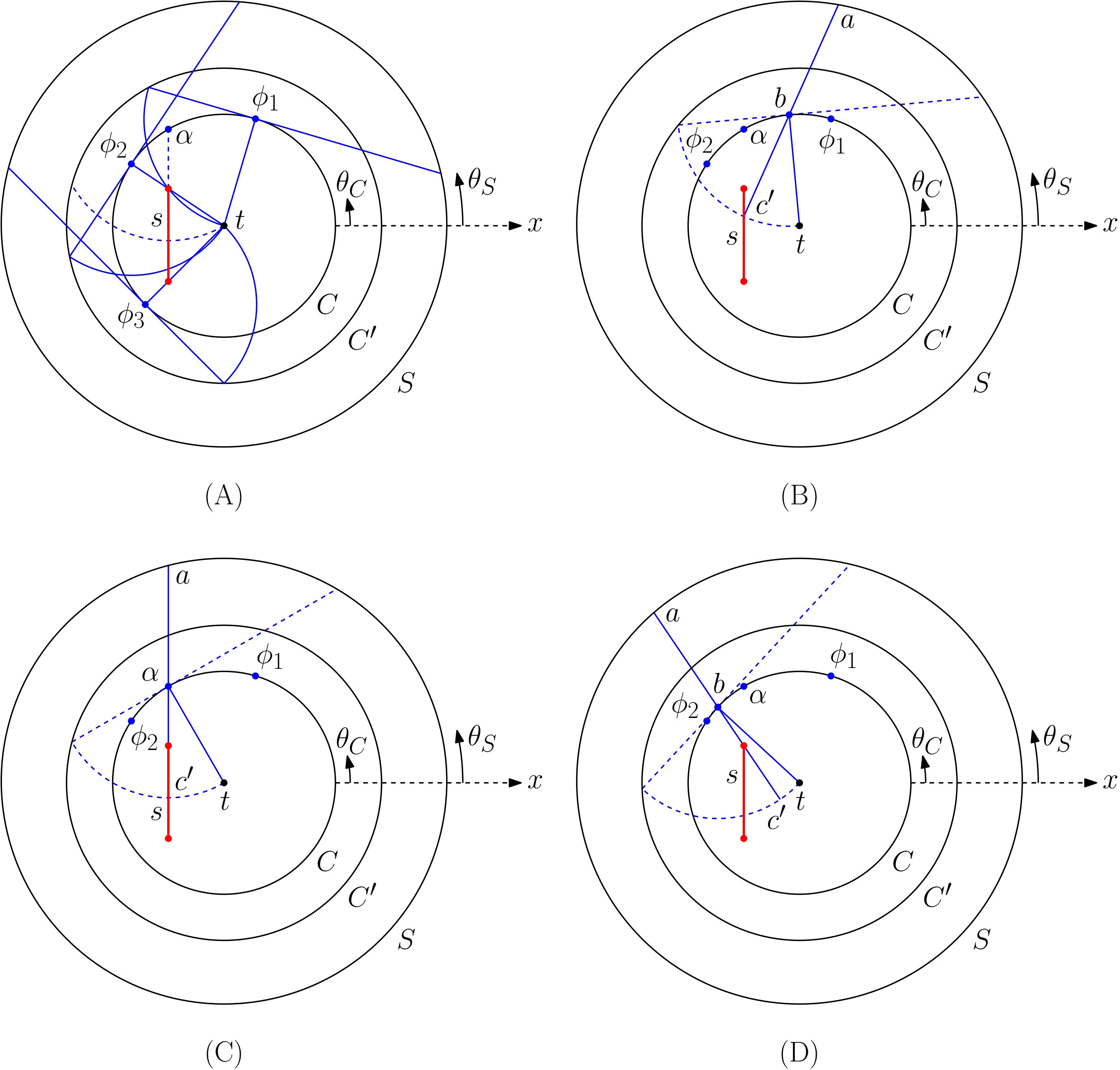}
	\caption{Infeasible final configurations due to an obstacle line segment $s$ inside $C$.
		Illustrations of $\theta_S$-lower bounds for (A) $\theta_C \in [\phi_1, \phi_2]$, (B) $\phi_1 < \theta_C < \alpha$, (C) $\theta_C = \alpha$, and (D) $\alpha < \theta_C < \phi_2$.}
	\label{fig_inf_in}
\end{figure}

We define $\phi_1$, $\phi_2$, and $\phi_3$ as follows (Figure \ref{fig_inf_in}A).
Angle $\phi_1$ is the smallest angle $\theta_C$ at which the circular arc of the quart-sector of $\theta_C$ intersects $s$ (at one of its endpoints or interior points).
Angle $\phi_2$ is the smallest angle $\theta_C$ at which $bt$ of the quart-sector of $\theta_C$ intersects $s$ (at one of its endpoints).
Angle $\phi_3$ is the largest angle $\theta_C$ at which $bt$ of the quart-sector of $\theta_C$ intersects $s$ (at one of its endpoints).
Observe that, as $\theta_C$ varies from 0 to $2\pi$, $\phi_1$ and $\phi_3$ are the angles $\theta_C$ at which the quart-sector of $\theta_C$ first and last intersects $s$, respectively.

We are only concerned with finding the lower bound of $\theta_S$ for $\theta_C \in [\phi_1, \phi_2]$, since the entire negative half of the $\theta_S$-range (i.e., $[\theta_C - \cos^{-1} r/R, \theta_C]$) is feasible for $\theta_C \in [0, \phi_1] \cup [\phi_3, 2\pi)$, and is infeasible for $\theta_C \in [\phi_2, \phi_3]$ due to intersection of $bt$ with $s$ (Figure \ref{fig_inf_in}A).

For $\theta_C \in [\phi_1, \phi_2]$, the lower bound of $\theta_S$ can be represented by a piecewise continuous curve, which consists of at most two pieces, corresponding to two intervals $[\phi_1, \alpha]$ and $[\alpha, \phi_2]$, where $\alpha$ is the angle $\theta_C$ of the intersection point between $C$ and the supporting line of $s$.
If $\phi_1 \leq \alpha$, then the curve has two pieces; otherwise, the curve is of one single piece.

For $\theta_C \in [\phi_1, \alpha]$, the lower bound of $\theta_S$ is indicated by the endpoint $a$ of line segment $abc'$,
%(i.e., intermediate configuration)
where $c'$ is the intersection point between $s$ and the circular arc centered at $b$ (Figure \ref{fig_inf_in}B).
If no intersection occurs between $s$ and the circular arc, then the lower bound of $\theta_S$ is given by the endpoint $a$ of line segment $abc'$, where $bc'$ passes through an endpoint of $s$.

For $\theta_C \in [\alpha, \phi_2]$, the lower bound of $\theta_S$ is indicated by the endpoint $a$ of line segment $abc'$, where $bc'$ passes through an endpoint of $s$ (Figure \ref{fig_inf_in}D).
%Observe that 
The lower bound of $\theta_S$ is equal to $\theta_C$ when $\theta_C = \phi_2$.
See Figure \ref{fig_inf_in_plot} for a sketch of the infeasible final configuration space.

\begin{figure}
	\centering
	\includegraphics[scale=0.135]{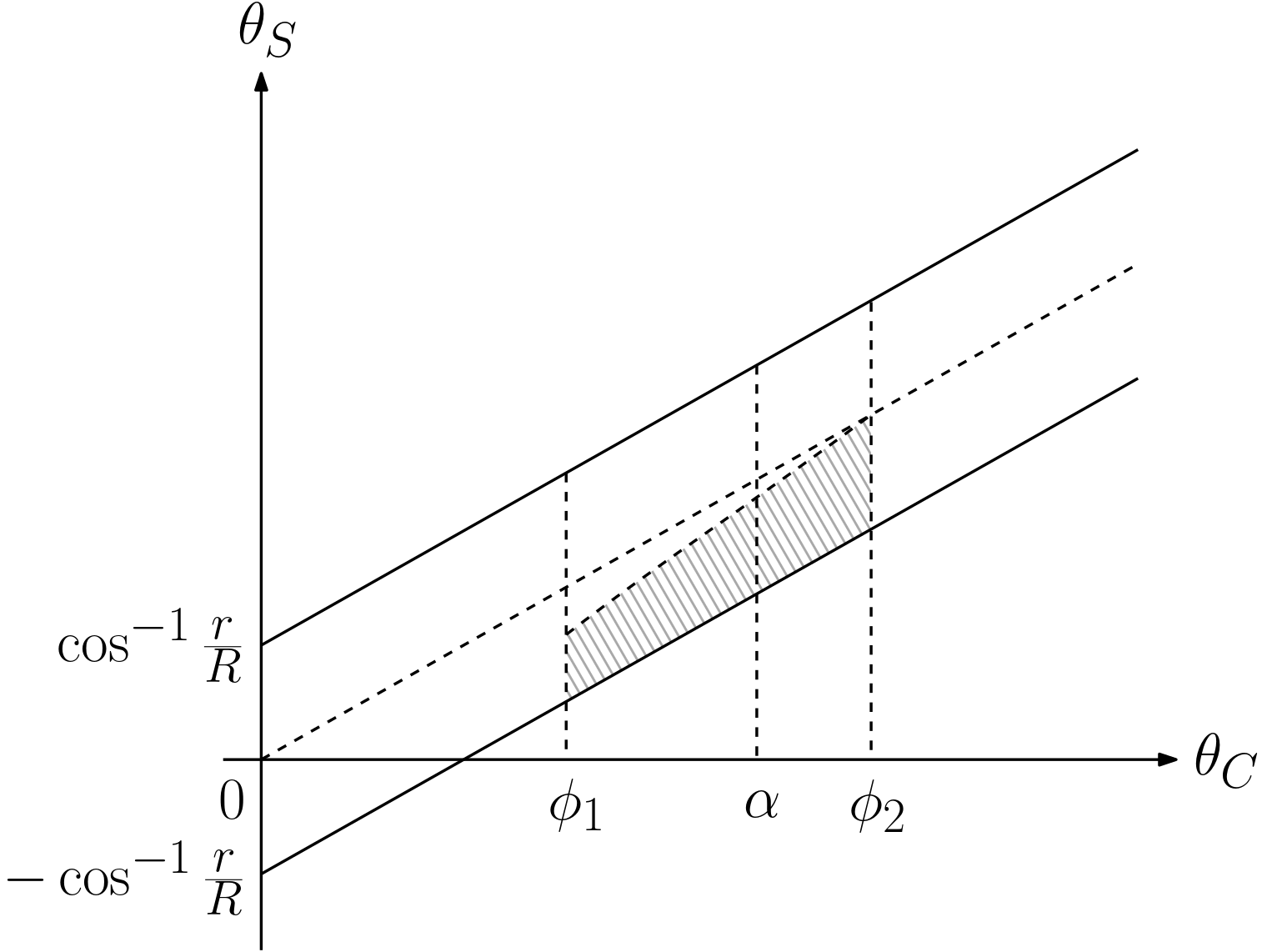}
	\caption{Infeasible final configuration space due to an obstacle line segment $s$ inside $C$.}
	\label{fig_inf_in_plot}
\end{figure}

\paragraph{Case 2. Obstacle line segment \boldmath$s$ outside \boldmath$C$ and inside \boldmath$C'$}
As depicted in Figure \ref{fig_inf_out}, we only need to worry about computing the lower bound of $\theta_S$ for $\theta_C \in [\phi_1, \phi_2]$, given that the entire negative half of the $\theta_S$-range (i.e., $[\theta_C - \cos^{-1} r/R, \theta_C]$) is feasible for $\theta_C \in [0, \phi_1] \cup [\phi_2, 2\pi)$. 
%The analysis is similar to case 1 and thus omitted herein.

\begin{figure}
	\centering
	\includegraphics[scale=0.135]{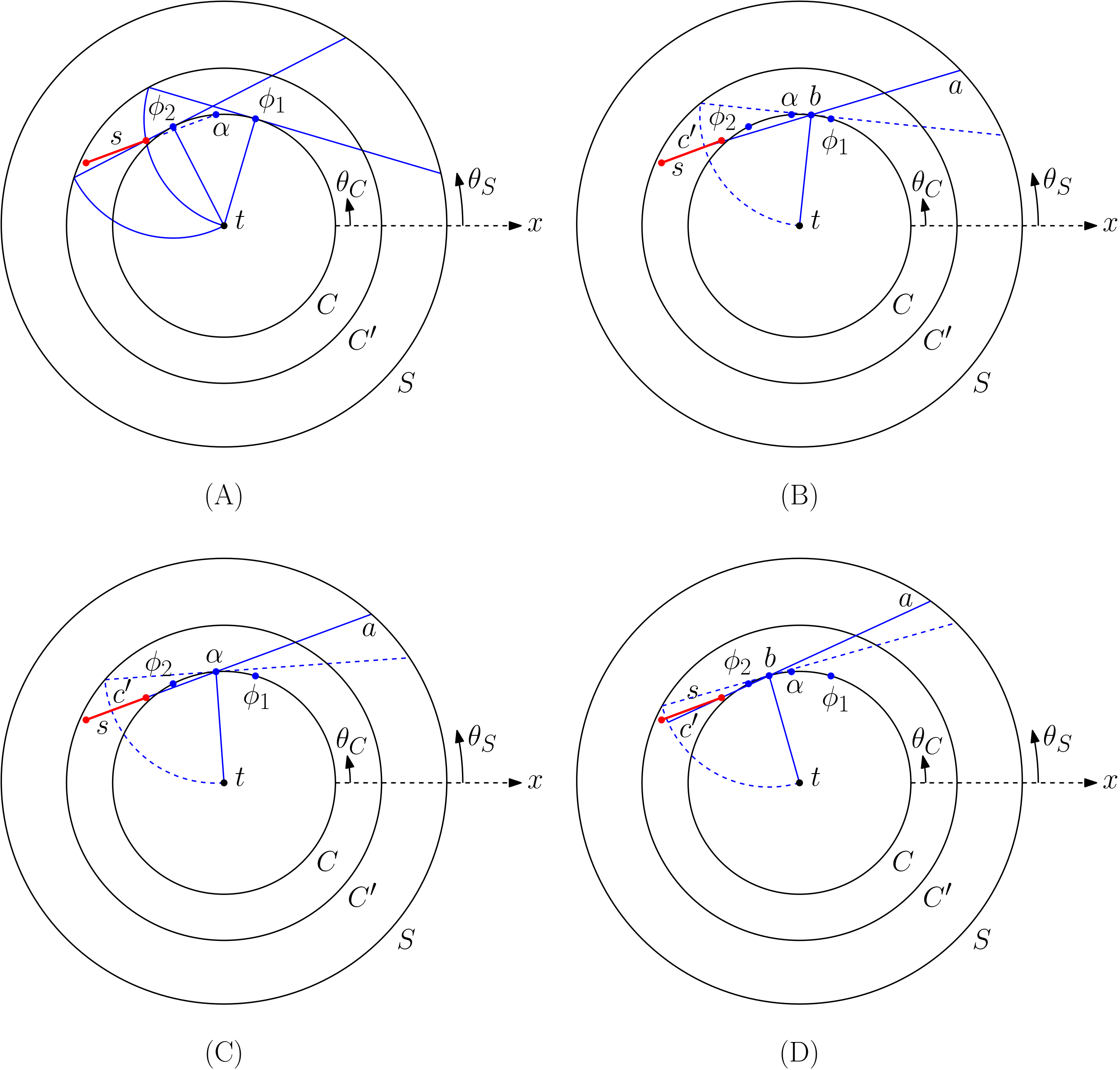}
	\caption{Infeasible final configurations due to an obstacle line segment $s$ outside $C$ and inside $C'$.
		Illustrations of $\theta_S$-lower bounds for (A) $\theta_C \in [\phi_1, \phi_2]$, (B) $\phi_1 < \theta_C < \alpha$, (C) $\theta_C = \alpha$, and (D) $\alpha < \theta_C < \phi_2$.}
	\label{fig_inf_out}
\end{figure}

The lower bound of $\theta_S$ for $\theta_C \in [\phi_1, \phi_2]$ is characterized by a piecewise continuous curve, which consists of at most two pieces, corresponding to intervals $[\phi_1, \alpha]$ and $[\alpha, \phi_2]$, where $\alpha$ is the angle $\theta_C$ of the intersection point, if any, between $C$ and the supporting line of $s$. 
If $\alpha$ exists, then the curve of $\theta_S$ has two pieces; otherwise, the curve of $\theta_S$ is a single piece.
Notice that $\phi_2$ is defined differently from the previous case.
Here, $\phi_2$ is the largest $\theta_C$ at which $bc'$ of the quart-sector of $\theta_C$ intersects with $s$ (at one of its endpoints).
In fact, when $\theta_C = \phi_2$, the lower bound of $\theta_S$ is equivalent to $\theta_C - \cos^{-1} r/R$.
A sketch of the corresponding infeasible final configuration space is shown in Figure \ref{fig_inf_out_plot}. \\

\begin{figure}
	\centering
	\includegraphics[scale=0.135]{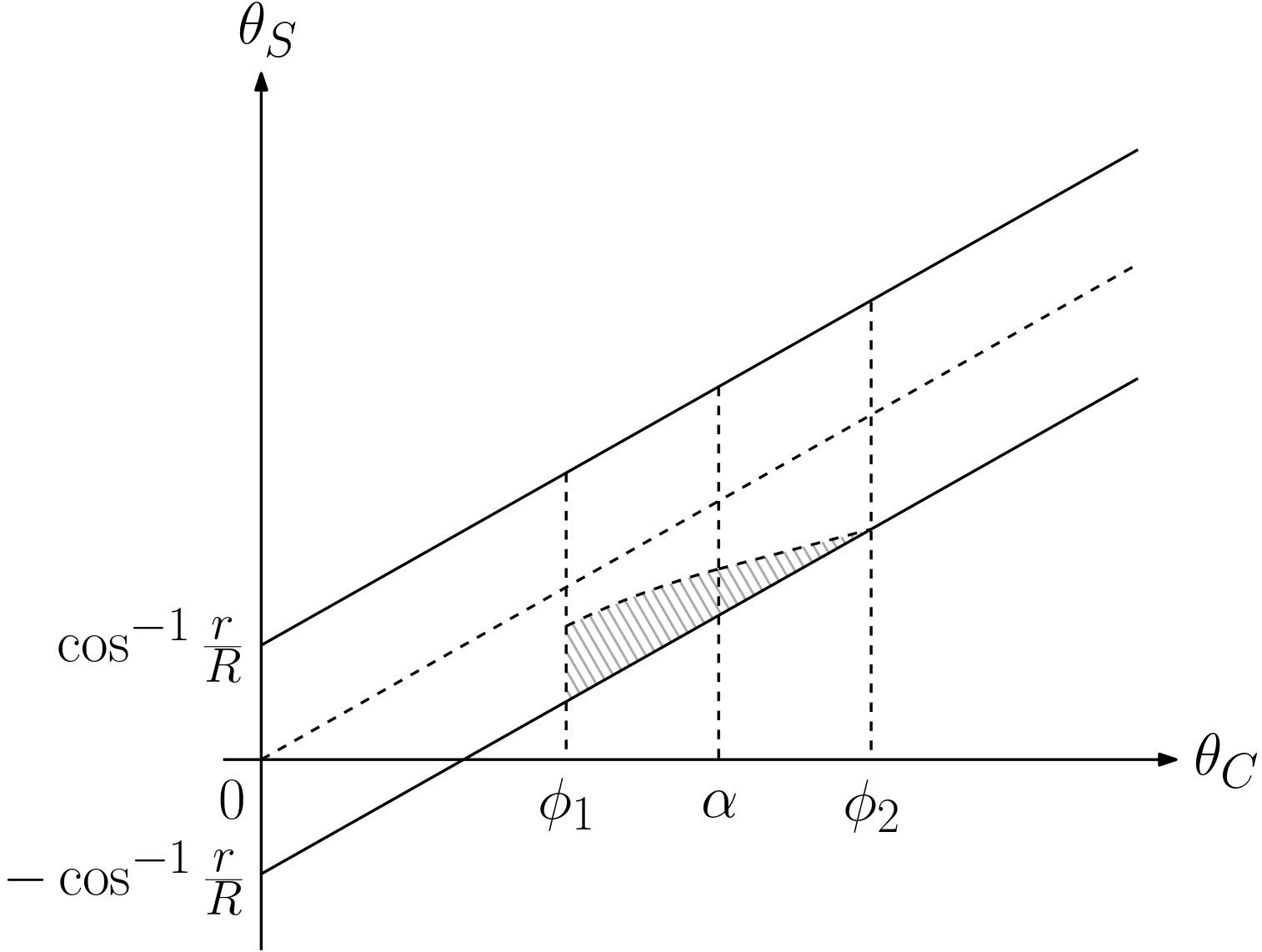}
	\caption{Infeasible final configuration space due to a line segment $s$ outside $C$ and inside $C'$.}
	\label{fig_inf_out_plot}
\end{figure}

Observe that any of the curves just described for characterizing the lower or upper bound of $\theta_S$ can be computed in constant time.
Thus, given an obstacle line segment $s$, the associated infeasible final configuration space can be found in $O(1)$ time.
As a result, it takes $O(n)$ time to determine the infeasible final configuration space for $n$ obstacle line segments.

\subsection{Complexity and construction of feasible trajectory space}
The feasible trajectory space of the articulated probe is represented by $\Sigma_{fin} \setminus (\Sigma_{fin, forb} \cup \Sigma_{fin, inf})$.
Three sets of lower- and upper-bound curves, denoted as $\sigma_{fin}$, $\sigma_{fin, forb}$, and $\sigma_{fin, inf}$, were obtained from characterizing the final, forbidden final, and infeasible final configuration spaces, respectively.
Each of these curves is a function of $\theta_S$ -- that is, $\theta_C(\theta_S)$.

As illustrated in Figure \ref{fig_fin}, $\sigma_{fin}$ contains two linearly increasing curves, $\theta_C = \theta_S - \cos^{-1} r/R$ and $\theta_C = \theta_S + \cos^{-1} r/R$, which are defined over $\theta_S \in [0, 2\pi)$.
Each curve in $\sigma_{fin, forb}$ is partially defined, continuous, and monotone in $\theta_S$.
Specifically, as shown in Figures \ref{fig_forb_out} and \ref{fig_forb_in}, the curves in case 1 are monotonically decreasing with respect to $\theta_S$, and the curves in case 2 are horizontal lines parallel to the $\theta_S$-axis (i.e., of some constant values of $\theta_C$).
Furthermore, any two curves in case 1 can intersect at most once.
Likewise, a curve in $\sigma_{fin, inf}$ is bounded and monotonically increasing with respect to $\theta_S$ (Figures \ref{fig_inf_in_plot} and \ref{fig_inf_out_plot}), and
%Any curve in $\sigma_{fin, inf}$ 
can intersect with another at most once.

From the observations above, it can be easily deduced that the number of intersections between any two curves in $\sigma = \sigma_{fin} \cup \sigma_{fin, forb} \cup \sigma_{fin, inf}$ is at most one.
%In other words, the curves of $\sigma$ are essentially lines, line segments, or pseudo-line segments.
For a set $\sigma$ of $O(n)$ $x$-monotone Jordan arcs, with at most $c$ intersections per pair of arcs, where $c$ is a constant, the maximum combinatorial complexity of the arrangement $A(\sigma)$ is $O(n^2)$ \cite{halperin97arrange}.

An incremental construction approach, as detailed in \cite{edelsbrunner92arrange}, can be used to construct the arrangement $A(\sigma)$ in $O(n^2\alpha(n))$
%$O(n \lambda_{c+2}(n))$
time using $O(n^2)$ space, where $\alpha(n)$ is the inverse Ackermann function.
%Since $c = 1$ in our case, the time required by the incremental algorithm is $O(n \lambda_3(n)) = O(n (n \alpha(n))) = O(n^2\alpha(n))$.
%A plane sweep algorithm can also be applied to constructing the arrangement.
%Its worst-case running time is slightly inferior to that of the incremental construction, even though it is output sensitive.
%By using the sweep paradigm, arrangement $A(\sigma)$ can be constructed in $O((n + k) \log n)$ time using $O(n + k)$ space, where $k$ is the number of intersection points in the arrangement, and $k = O(n^2)$.
By using topological sweep \cite{balaban95optimal} in computing the intersections for a collection of well-behaved curves such as those described above, the time and space complexities can be improved to $O(n \log n + k)$ and $O(n + k)$, respectively.
Note that we can find a feasible probe trajectory by simply traversing the cells of the arrangement $A(\sigma)$ in $O(n^2)$ time.
This implies an $O(\log n)$ improvement over the previous result reported in \cite{teo20traj}.
We thus conclude with the following theorem. 

\begin{thm}
	\label{thm:feas_space}
	For a positive value $r$, the feasible trajectory space of the corresponding articulated probe can be represented as a simple arrangement of maximum combinatorial complexity $k = O(n^2)$, which can be constructed in $O(n \log n + k)$ time using $O(n + k)$ space.
	A feasible probe trajectory, if one exists, can be determined in $O(n^2)$ time using $O(n^2)$ space.
\end{thm}

\section{Conclusion}
\label{conc}

We have presented efficient geometric-combinatorial algorithms for finding feasible trajectories for a simple articulated probe, whose end link has a customizable length $r$.
Specifically, we can determine, in $O(n^{2+\epsilon})$ time, the minimum length $r$ for which a feasible articulated probe trajectory exists, and report at least one such trajectory.
In addition, we can compute the entire feasible domain of $r$ in $O(n^{5/2})$ time using $O(n^{2+\epsilon})$ space.
We can also characterize the feasible trajectory space of the articulated probe, for a given $r$, by using an arrangement whose complexity is $O(k)$, and its construction takes $O(n \log n + k)$ time and $O(n + k)$ space, where $k = O(n^2)$ is the number of vertices of the arrangement.
We can find, using the arrangement, a feasible probe trajectory for a given $r$ in $O(n^2)$ time.

A number of open problems remain:
1) Our solution to Problem \ref{prob_min_r} relies on efficient data structures to address some rather specific geometric intersection and emptiness query problems.
Can we improve upon those query data structures?
2) Do our techniques extend well to the variant in which a clearance is required from the obstacles?
3) Can we generalize our solution approaches to three dimensions?

%% If you have bibdatabase file and want bibtex to generate the
%% bibitems, please use
%%
%\bibliographystyle{elsarticle-num-names}
\bibliographystyle{elsarticle-harv}
\bibliography{refs}

\begin{thebibliography}{25}
\expandafter\ifx\csname natexlab\endcsname\relax\def\natexlab#1{#1}\fi
\providecommand{\url}[1]{\texttt{#1}}
\providecommand{\href}[2]{#2}
\providecommand{\path}[1]{#1}
\providecommand{\DOIprefix}{doi:}
\providecommand{\ArXivprefix}{arXiv:}
\providecommand{\URLprefix}{URL: }
\providecommand{\Pubmedprefix}{pmid:}
\providecommand{\doi}[1]{\href{http://dx.doi.org/#1}{\path{#1}}}
\providecommand{\Pubmed}[1]{\href{pmid:#1}{\path{#1}}}
\providecommand{\bibinfo}[2]{#2}
\ifx\xfnm\relax \def\xfnm[#1]{\unskip,\space#1}\fi
%Type = Article
\bibitem[{Agarwal et~al.(1996)Agarwal, Schwarzkopf and
  Sharir}]{agarwal96overlay}
\bibinfo{author}{Agarwal, P.K.}, \bibinfo{author}{Schwarzkopf, O.},
  \bibinfo{author}{Sharir, M.}, \bibinfo{year}{1996}.
\newblock \bibinfo{title}{The overlay of lower envelopes and its applications}.
\newblock \bibinfo{journal}{Discrete \& Computational Geometry}
  \bibinfo{volume}{15}, \bibinfo{pages}{1--13}.
%Type = Techreport
\bibitem[{Arkin and Mitchell(1987)}]{arkin87opt}
\bibinfo{author}{Arkin, E.}, \bibinfo{author}{Mitchell, J.},
  \bibinfo{year}{1987}.
\newblock \bibinfo{title}{An optimal visibility algorithm for a simple polygon
  with star-shaped holes}.
\newblock \bibinfo{type}{Technical Report}. Cornell University Operations
  Research and Industrial Engineering.
%Type = Inproceedings
\bibitem[{Balaban(1995)}]{balaban95optimal}
\bibinfo{author}{Balaban, I.J.}, \bibinfo{year}{1995}.
\newblock \bibinfo{title}{An optimal algorithm for finding segments
  intersections}, in: \bibinfo{booktitle}{Proceedings of the eleventh annual
  symposium on Computational geometry}, pp. \bibinfo{pages}{211--219}.
%Type = Article
\bibitem[{Brooks and Lozano-Perez(1985)}]{brooks85subdivision}
\bibinfo{author}{Brooks, R.A.}, \bibinfo{author}{Lozano-Perez, T.},
  \bibinfo{year}{1985}.
\newblock \bibinfo{title}{A subdivision algorithm in configuration space for
  findpath with rotation}.
\newblock \bibinfo{journal}{IEEE Transactions on Systems, Man, and Cybernetics}
  \bibinfo{volume}{SMC-15}, \bibinfo{pages}{224--233}.
%Type = Book
\bibitem[{Choset et~al.(2005)Choset, Hutchinson, Lynch, Kantor, Burgard,
  Kavraki and Thrun}]{choset05principles}
\bibinfo{author}{Choset, H.M.}, \bibinfo{author}{Hutchinson, S.},
  \bibinfo{author}{Lynch, K.M.}, \bibinfo{author}{Kantor, G.},
  \bibinfo{author}{Burgard, W.}, \bibinfo{author}{Kavraki, L.E.},
  \bibinfo{author}{Thrun, S.}, \bibinfo{year}{2005}.
\newblock \bibinfo{title}{Principles of robot motion: theory, algorithms, and
  implementation}.
\newblock \bibinfo{publisher}{MIT Press}.
%Type = Article
\bibitem[{Connelly and Demaine(2017)}]{connelly17geom}
\bibinfo{author}{Connelly, R.}, \bibinfo{author}{Demaine, E.D.},
  \bibinfo{year}{2017}.
\newblock \bibinfo{title}{Geometry and topology of polygonal linkages}.
\newblock \bibinfo{journal}{Handbook of Discrete and Computational Geometry} ,
  \bibinfo{pages}{233--256}.
%Type = Article
\bibitem[{Culmone et~al.(2019)Culmone, Smit and Breedveld}]{culmone19additive}
\bibinfo{author}{Culmone, C.}, \bibinfo{author}{Smit, G.},
  \bibinfo{author}{Breedveld, P.}, \bibinfo{year}{2019}.
\newblock \bibinfo{title}{Additive manufacturing of medical instruments: A
  state-of-the-art review}.
\newblock \bibinfo{journal}{Additive Manufacturing} .
%Type = Inproceedings
\bibitem[{Daescu and Teo(2019)}]{daescu19traj3d}
\bibinfo{author}{Daescu, O.}, \bibinfo{author}{Teo, K.Y.},
  \bibinfo{year}{2019}.
\newblock \bibinfo{title}{Computing feasible trajectories for an articulated
  probe in three dimensions}, in: \bibinfo{booktitle}{31st Annual Canadian
  Conference on Computational Geometry}, pp. \bibinfo{pages}{59--70}.
%Type = Techreport
\bibitem[{Donald(1984)}]{donald84motion}
\bibinfo{author}{Donald, B.R.}, \bibinfo{year}{1984}.
\newblock \bibinfo{title}{Motion planning with six degrees of freedom}.
\newblock \bibinfo{type}{Technical Report}. MIT Artificial Intelligence Lab.
%Type = Article
\bibitem[{Edelsbrunner et~al.(1992)Edelsbrunner, Guibas, Pach, Pollack, Seidel
  and Sharir}]{edelsbrunner92arrange}
\bibinfo{author}{Edelsbrunner, H.}, \bibinfo{author}{Guibas, L.},
  \bibinfo{author}{Pach, J.}, \bibinfo{author}{Pollack, R.},
  \bibinfo{author}{Seidel, R.}, \bibinfo{author}{Sharir, M.},
  \bibinfo{year}{1992}.
\newblock \bibinfo{title}{Arrangements of curves in the plane -- topology,
  combinatorics, and algorithms}.
\newblock \bibinfo{journal}{Theoretical Computer Science} \bibinfo{volume}{92},
  \bibinfo{pages}{319--336}.
%Type = Article
\bibitem[{Halperin and Sharir(2017)}]{halperin97arrange}
\bibinfo{author}{Halperin, D.}, \bibinfo{author}{Sharir, M.},
  \bibinfo{year}{2017}.
\newblock \bibinfo{title}{Arrangements}.
\newblock \bibinfo{journal}{Handbook of Discrete and Computational Geometry} ,
  \bibinfo{pages}{723--762}.
%Type = Article
\bibitem[{Heffernan and Mitchell(1995)}]{heffernan95opt}
\bibinfo{author}{Heffernan, P.J.}, \bibinfo{author}{Mitchell, J.S.},
  \bibinfo{year}{1995}.
\newblock \bibinfo{title}{An optimal algorithm for computing visibility in the
  plane}.
\newblock \bibinfo{journal}{SIAM Journal on Computing} \bibinfo{volume}{24},
  \bibinfo{pages}{184--201}.
%Type = Article
\bibitem[{Hershberger(1989)}]{Hersh89find}
\bibinfo{author}{Hershberger, J.}, \bibinfo{year}{1989}.
\newblock \bibinfo{title}{Finding the upper envelope of \textit{n} line
  segments in \textit{O}(\textit{n} log \textit{n}) time}.
\newblock \bibinfo{journal}{Information Processing Letters}
  \bibinfo{volume}{33}, \bibinfo{pages}{169--174}.
%Type = Article
\bibitem[{Hopcroft et~al.(1984)Hopcroft, Joseph and Whitesides}]{hopcroft84mov}
\bibinfo{author}{Hopcroft, J.}, \bibinfo{author}{Joseph, D.},
  \bibinfo{author}{Whitesides, S.}, \bibinfo{year}{1984}.
\newblock \bibinfo{title}{Movement problems for 2-dimensional linkages}.
\newblock \bibinfo{journal}{SIAM Journal on Computing} \bibinfo{volume}{13},
  \bibinfo{pages}{610--629}.
%Type = Article
\bibitem[{Kavraki et~al.(1996)Kavraki, Svestka, Latombe and
  Overmars}]{kavraki96probabilistic}
\bibinfo{author}{Kavraki, L.}, \bibinfo{author}{Svestka, P.},
  \bibinfo{author}{Latombe, J.}, \bibinfo{author}{Overmars, M.},
  \bibinfo{year}{1996}.
\newblock \bibinfo{title}{Probabilistic roadmaps for path planning in
  high-dimensional configuration spaces}.
\newblock \bibinfo{journal}{IEEE Transactions on Robotics and Automation}
  \bibinfo{volume}{12}, \bibinfo{pages}{566--580}.
%Type = Book
\bibitem[{LaValle(2006)}]{lavalle06plan}
\bibinfo{author}{LaValle, S.M.}, \bibinfo{year}{2006}.
\newblock \bibinfo{title}{Planning algorithms}.
\newblock \bibinfo{publisher}{Cambridge University Press}.
%Type = Article
\bibitem[{LaValle and Kuffner~Jr(2001)}]{lavalle01randomized}
\bibinfo{author}{LaValle, S.M.}, \bibinfo{author}{Kuffner~Jr, J.J.},
  \bibinfo{year}{2001}.
\newblock \bibinfo{title}{Randomized kinodynamic planning}.
\newblock \bibinfo{journal}{The International Journal of Robotics Research}
  \bibinfo{volume}{20}, \bibinfo{pages}{378--400}.
%Type = Article
\bibitem[{Matou{\v{s}}ek(1993)}]{matouvsek93range}
\bibinfo{author}{Matou{\v{s}}ek, J.}, \bibinfo{year}{1993}.
\newblock \bibinfo{title}{Range searching with efficient hierarchical
  cuttings}.
\newblock \bibinfo{journal}{Discrete \& Computational Geometry}
  \bibinfo{volume}{10}, \bibinfo{pages}{157--182}.
%Type = Inproceedings
\bibitem[{Pocchiola(1990)}]{pocchiola90graphics}
\bibinfo{author}{Pocchiola, M.}, \bibinfo{year}{1990}.
\newblock \bibinfo{title}{Graphics in flatland revisited}, in:
  \bibinfo{booktitle}{Scandinavian Workshop on Algorithm Theory}, pp.
  \bibinfo{pages}{85--96}.
%Type = Book
\bibitem[{Sharir and Agarwal(1995)}]{sharir95dav}
\bibinfo{author}{Sharir, M.}, \bibinfo{author}{Agarwal, P.K.},
  \bibinfo{year}{1995}.
\newblock \bibinfo{title}{{Davenport-Schinzel} sequences and their geometric
  applications}.
\newblock \bibinfo{publisher}{Cambridge University Press}.
%Type = Article
\bibitem[{Simaan et~al.(2018)Simaan, Yasin and Wang}]{simaan18medical}
\bibinfo{author}{Simaan, N.}, \bibinfo{author}{Yasin, R.M.},
  \bibinfo{author}{Wang, L.}, \bibinfo{year}{2018}.
\newblock \bibinfo{title}{Medical technologies and challenges of robot-assisted
  minimally invasive intervention and diagnostics}.
\newblock \bibinfo{journal}{Annual Review of Control, Robotics, and Autonomous
  Systems} \bibinfo{volume}{1}, \bibinfo{pages}{465--490}.
%Type = Inproceedings
\bibitem[{Suri and O'Rourke(1986)}]{suri86worst}
\bibinfo{author}{Suri, S.}, \bibinfo{author}{O'Rourke, J.},
  \bibinfo{year}{1986}.
\newblock \bibinfo{title}{Worst-case optimal algorithms for constructing
  visibility polygons with holes}, in: \bibinfo{booktitle}{Proceedings of the
  Second Annual Symposium on Computational Geometry},
  \bibinfo{organization}{ACM}. pp. \bibinfo{pages}{14--23}.
%Type = Article
\bibitem[{Teo et~al.(2020)Teo, Daescu and Fox}]{teo20traj}
\bibinfo{author}{Teo, K.Y.}, \bibinfo{author}{Daescu, O.},
  \bibinfo{author}{Fox, K.}, \bibinfo{year}{2020}.
\newblock \bibinfo{title}{Trajectory planning for an articulated probe}.
\newblock \bibinfo{journal}{Computational Geometry} , \bibinfo{pages}{101655}.
%Type = Article
\bibitem[{Yakey et~al.(2001)Yakey, LaValle and Kavraki}]{yakey01randomized}
\bibinfo{author}{Yakey, J.H.}, \bibinfo{author}{LaValle, S.M.},
  \bibinfo{author}{Kavraki, L.E.}, \bibinfo{year}{2001}.
\newblock \bibinfo{title}{Randomized path planning for linkages with closed
  kinematic chains}.
\newblock \bibinfo{journal}{IEEE Transactions on Robotics and Automation}
  \bibinfo{volume}{17}, \bibinfo{pages}{951--958}.
%Type = Inproceedings
\bibitem[{Zhu and Latombe(1990)}]{zhu90constraint}
\bibinfo{author}{Zhu, D.}, \bibinfo{author}{Latombe, J.}, \bibinfo{year}{1990}.
\newblock \bibinfo{title}{Constraint reformulation in a hierarchical path
  planner}, in: \bibinfo{booktitle}{Proceedings of the IEEE International
  Conference on Robotics and Automation}, \bibinfo{organization}{IEEE}. pp.
  \bibinfo{pages}{1918--1923}.

\end{thebibliography}

\end{document}